\documentclass[12pt, twocolumnfalse]{IEEEtran}

\usepackage{epsfig,graphics,amssymb,color,rotating}
\usepackage{amscd}
\usepackage{amsmath}
\usepackage{enumerate,multirow}
\usepackage{theorem}
\usepackage{epsfig}
\usepackage{verbatim}

\definecolor{blau}{rgb}{.15,.4,.8}
\definecolor{gruen}{rgb}{.25,.8,.25}
\definecolor{rot}{rgb}{.8,0,0}
\definecolor{gelb}{rgb}{.8,.8,0}
\definecolor{lila}{rgb}{.8,0,.8}
\definecolor{flieder}{rgb}{.77,.79,.91}
\definecolor{beige}{rgb}{.96,.79,.48}
\definecolor{orange}{rgb}{.8,.4,.1}
\definecolor{lightgreen}{rgb}{.73,.95,.76}
\definecolor{lightred}{rgb}{1,.71,.67}
\definecolor{lightblue}{rgb}{.52,.83,.98}

\definecolor{white}{rgb}{1,1,1}
\definecolor{blue}{rgb}{.15,.4,.8}
\definecolor{green}{rgb}{0,.8,.2}
\definecolor{red}{rgb}{.8,0,0}
\definecolor{yellow}{rgb}{.8,.8,0}
\definecolor{purple}{rgb}{.8,0,.8}

\newcommand{\eq}[1]{(\ref{#1})}
\newcommand{\be}{\begin{equation}}
\newcommand{\ee}{\end{equation}}
\newcommand{\ba}{\begin{array}}
\newcommand{\ea}{\end{array}}
\newcommand{\bea}{\begin{eqnarray}}
\newcommand{\eea}{\end{eqnarray}}

\def\mathlette#1#2{{\mathchoice{\mbox{#1$\displaystyle #2$}}%
                               {\mbox{#1$\textstyle #2$}}%
                               {\mbox{#1$\scriptstyle #2$}}%
                               {\mbox{#1$\scriptscriptstyle #2$}}}}

\newcommand{\Mat}[1]{\ensuremath{\mathlette{\boldmath}{#1}}}
\renewcommand{\Vec}[1]{\ensuremath{\mathlette{\boldmath}{#1}}}

\providecommand{\K}{\ensuremath{K}}

\providecommand{\N}{\ensuremath{N}}

\providecommand{\matS}{\ensuremath{{\Mat{S}}}}
\providecommand{\matT}{\ensuremath{{\Mat{T}}}}

\providecommand{\matH}{\ensuremath{{\Mat{H}}}}

\providecommand{\K}{\ensuremath{K}}

\providecommand{\N}{\ensuremath{N}}

\providecommand{\matI}{\ensuremath{{\Mat{I}}}}

\providecommand{\matS}{\ensuremath{{\Mat{S}}}}

\providecommand{\matT}{\ensuremath{{\Mat{T}}}}

\providecommand{\matH}{\ensuremath{{\Mat{H}}}}

\providecommand{\matA}{\ensuremath{{\Mat{A}}}}
\providecommand{\matC}{\ensuremath{{\Mat{C}}}}

\providecommand{\matXi}{\ensuremath{{\boldsymbol{\Xi}}}}
\providecommand{\matQTilde}{\ensuremath{{\Mat{\widetilde{Q}}}}}
\providecommand{\matPTilde}{\ensuremath{{\Mat{\widetilde{P}}}}}
\providecommand{\matGTilde}{\ensuremath{{\Mat{\widetilde{G}}}}}

\providecommand{\matTTilde}{\ensuremath{{\Mat{\widetilde{T}}}}}
\providecommand{\matRTilde}{\ensuremath{{\Mat{\widetilde{R}}}}}

\providecommand{\matG}{\ensuremath{{\Mat{G}}}}
\providecommand{\matP}{\ensuremath{{\Mat{P}}}}

\providecommand{\stackH}{\ensuremath{{\boldsymbol{\mathcal{H}}}}}

\providecommand{\stacky}{\ensuremath{{\boldsymbol{\mathcal{Y}}}}}
\providecommand{\stackb}{\ensuremath{{\boldsymbol{\mathcal{B}}}}}

\providecommand{\stackW}{\ensuremath{{\boldsymbol{\mathcal{W}}}}}

\providecommand{\stackT}{\ensuremath{{\boldsymbol{\mathcal{T}}}}}
\providecommand{\stackR}{\ensuremath{{\boldsymbol{\mathcal{R}}}}}

\providecommand{\vecb}{\ensuremath{{\Vec{b}}}}

\providecommand{\vecy}{\ensuremath{{\Vec{y}}}}

\providecommand{\vecs}{\ensuremath{{\Vec{s}}}}
\providecommand{\vech}{\ensuremath{{\Vec{h}}}}

\newcommand{\I}{\Mat{I}}
\newcommand{\E}{\mathrm{E}}

\newtheorem{theor}{Theorem}
\newtheorem{proposition}{Proposition}
\newtheorem{corollary}{Corollary}

\newtheorem{lemma}{Lemma}

\newtheorem{definition}{Definition}
\newtheorem{property}{Property}

\newcommand{\mnov}{Nov}

\newcommand{\usao}{U.S.A.} 
\newcommand{\usa}{\usao, }

\makeatletter
\def\ifundefined{\@ifundefined}
\makeatother

\renewcommand{\baselinestretch}{1.66}
\begin{document}

\title{Asynchronous
CDMA Systems with Random Spreading--Part I: Fundamental Limits}

\author{Laura Cottatellucci, Ralf R. M\"uller, and M\'erouane Debbah
\thanks{This work was presented in part at the 42$^{th}$ Annual Asilomar
Conference on Signals, Systems, and Computers,  Pacific Grove,
CA, \usa \mnov \ 2007. }
\thanks{This work was supported in part by the  French ANR "Masses de Données" project SESAME  and by the Research Council of Norway under grant 171133/V30. }
\thanks{Laura Cottatellucci is with Eurecom, Sophia Antipolis, France (e-mail:
laura.cottatellucci@eurecom.fr). She was with Institute of
Telecommunications Research, University of South Australia,
Adealide, SA, Australia. Ralf M\"uller is with Norwegian
University of Science and Technology, Trondheim, Norway, (e-mail:
mueller@iet.ntnu.no). M\'erouane Debbah was with Eurecom, Sophia Antipolis, France. He is currently with SUPELEC,  91192 Gif-sur-Yvette, France (e-mail: merouane.debbah@supelec.fr).} }

\ifundefined{IEEEtransversionmajor}{%
  \newlength{\IEEEilabelindent}
  \newlength{\IEEEilabelindentA}
  \newlength{\IEEEilabelindentB}
  \newlength{\IEEEelabelindent}
  \newlength{\IEEEdlabelindent}
  \newlength{\labelindent}
  \newlength{\IEEEiednormlabelsep}
  \newlength{\IEEEiedmathlabelsep}
  \newlength{\IEEEiedtopsep}

  \providecommand{\IEEElabelindentfactori}{1.0}
  \providecommand{\IEEElabelindentfactorii}{0.75}
  \providecommand{\IEEElabelindentfactoriii}{0.0}
  \providecommand{\IEEElabelindentfactoriv}{0.0}
  \providecommand{\IEEElabelindentfactorv}{0.0}
  \providecommand{\IEEElabelindentfactorvi}{0.0}
  \providecommand{\labelindentfactor}{1.0}

  \providecommand{\iedlistdecl}{\relax}
  \providecommand{\calcleftmargin}[1]{
                  \setlength{\leftmargin}{#1}
                  \addtolength{\leftmargin}{\labelwidth}
                  \addtolength{\leftmargin}{\labelsep}}
  \providecommand{\setlabelwidth}[1]{
                  \settowidth{\labelwidth}{#1}}
  \providecommand{\usemathlabelsep}{\relax}
  \providecommand{\iedlabeljustifyl}{\relax}
  \providecommand{\iedlabeljustifyc}{\relax}
  \providecommand{\iedlabeljustifyr}{\relax}

  \newif\ifnocalcleftmargin
  \nocalcleftmarginfalse

  \newif\ifnolabelindentfactor
  \nolabelindentfactorfalse

  \newif\ifcenterfigcaptions
  \centerfigcaptionsfalse

  \let\OLDitemize\itemize
  \let\OLDenumerate\enumerate
  \let\OLDdescription\description

  \renewcommand{\itemize}[1][\relax]{\OLDitemize}
  \renewcommand{\enumerate}[1][\relax]{\OLDenumerate}
  \renewcommand{\description}[1][\relax]{\OLDdescription}

  \providecommand{\pubid}[1]{\relax}
  \providecommand{\pubidadjcol}{\relax}
  \providecommand{\specialpapernotice}[1]{\relax}
  \providecommand{\overrideIEEEmargins}{\relax}

    \let\CMPARstart\PARstart

  \let\OLDappendix\appendix
  \renewcommand{\appendix}[1][\relax]{\OLDappendix}

  \newif\ifuseRomanappendices
  \useRomanappendicestrue

   \let\OLDbiography\biography
  \let\OLDendbiography\endbiography
  \renewcommand{\biography}[2][\relax]{\OLDbiography{#2}}
  \renewcommand{\endbiography}{\OLDendbiography}

  \markboth{A Test for IEEEtran.cls--- {\tiny \bfseries
  [Running Older Class]}}{Shell: A Test for IEEEtran.cls}}{

    \markboth{}%
  {}}

\renewcommand{\baselinestretch}{1.4}

\maketitle

\begin{abstract}
Spectral efficiency for asynchronous code
division multiple access (CDMA) with random spreading is calculated in the large
system limit allowing for arbitrary chip waveforms and
frequency-flat fading. Signal to interference
and noise ratios (SINRs) for suboptimal receivers, such
as the linear minimum mean square error (MMSE) detectors, are derived. The approach is general
and optionally allows even for statistics obtained by
under-sampling the received signal.

All performance measures are given as a function of the chip waveform and the delay
distribution of the users in the large system limit.
It turns out that synchronizing users on a chip level impairs performance for all chip waveforms with bandwidth greater than the Nyquist bandwidth, e.g.,\ positive roll-off factors. For example, with the pulse shaping demanded in the UMTS standard, user synchronization reduces spectral efficiency up to 12\% at 10~dB normalized signal-to-noise ratio. The benefits of asynchronism stem from the finding that the excess bandwidth of chip waveforms actually spans additional dimensions in signal space, if the users are de-synchronized on the chip-level.

The analysis of linear MMSE detectors shows that the limiting
interference effects can be decoupled both in the user domain and
in the frequency domain such that the concept of the
effective interference spectral density arises. This generalizes and refines Tse and Hanly's
concept of effective interference.

In Part II, the analysis is extended to any linear detector that
admits a representation as multistage detector and guidelines for
the design of low complexity multistage detectors with universal
weights are provided.

\vspace{10mm} \emph{Index Terms} - Asynchronous code division multiple access (CDMA), channel capacity, effective interference, excess bandwidth, minimum mean square error (MMSE) detector,  multistage detector, multiuser detection,  pulse shaping, random matrix theory,  random spreading sequences, spectral efficiency.
\end{abstract}

\renewcommand{\baselinestretch}{2}
\vspace{0mm}

\section{Introduction}\label{chap:async_sec:introduction}
The fundamental limits of synchronous code-division multiple-access
(CDMA) systems and the loss incurred by the imposition of
suboptimal receiving structures have been thoroughly studied in
different scenarios and from different perspectives. On the one
hand, significant efforts have been devoted to characterize the
optimal spreading sequences and the corresponding capacities
\cite{rupf:94b,viswanath:99a,viswanath:99b}. On the other hand,
very insightful analysis
\cite{verdu:99a,tse:99b,mueller:98d,shitz:99,tanaka:02} resulted from
modelling the spreading sequences by random sequences
\cite{grant:98}. In fact, as both the number of users $K$   and
the spreading factor $N$ tend to infinity with a fixed ratio, CDMA
systems with random spreading show self-averaging properties.
These enable the description of the system in terms of few
macroscopic system parameters and thus provide a deep
understanding of the system behavior.

In the literature, the fundamental limits of CDMA systems and the
asymptotic analysis of linear multiuser detectors under the
assumption of random spreading sequences is overwhelmingly focused on
synchronous CDMA systems.  While the assumption of user synchronization
allowed for accurate large-system analysis, it is not realistic
for the received signal on the uplink of a cellular CDMA system,
in particular if users move and cause varying delays.
Therefore, it is of theoretical and practical interest to
extend the analysis of CDMA systems with random spreading to
asynchronous users. This holds in particular, as we will see that
asynchronous users are beneficial from a viewpoint of system performance.

The analysis of asynchronous CDMA systems using a single-user matched filter as
receiver was first given in \cite{pursley:77a}. A rich field of analysis of
asynchronous CDMA systems with conventional detection at the
receiver is based on Gaussian approximation methods. An exhaustive
overview of these approaches exceeds the scope of this work,
which is focused on the analysis of asynchronous CDMA systems with
{\em optimal joint decoding or linear multiuser detection}. The
interested reader is referred to \cite{zang:03} \cite{jeong:06}
and references therein for asynchronous CDMA with single-user receivers.

The analysis and design of asynchronous CDMA systems with linear
detectors is predominantly restricted to consider symbol-asynchronous
but chip-synchronous signals, i.e.,\ the time delays of the signals
are multiples of the chip interval. The effect of chip-asynchronism
is eventually analyzed independently
\cite{schramm:99}. In this stream are works that optimize the
spreading sequences to maximize the sum capacity \cite{luo:05} and
analyze the performance of linear multiuser detectors
\cite{schramm:99,kiran:00,cottatellucci:04b,cottatellucci:04g}. In
\cite{schramm:99,kiran:00}, the linear MMSE detector for
symbol-asynchronous but chip-synchronous systems is shown to attain the
performance of the linear MMSE detector for synchronous systems
as the size of the observation window tends to infinity by empirical
and analytical means, respectively. However,
they verify numerically that the performance of linear MMSE
detectors is severely impaired by the use of short observation
windows. Additionally, \cite{kiran:00} provides the large-system
SINR for a symbol whose chips are completely received in an
observation window of length equal to the symbol interval $T_s$.
In \cite{cottatellucci:04b,cottatellucci:04g}, the analysis of
linear multistage and MMSE detectors is extended to observation
windows of arbitrary length. Furthermore,  a multistage detector
structure that does not suffer from windowing effects and performs
as well as the multistage detector for synchronous systems is
proposed.

In \cite{schramm:99,mantravadi:02}, the effects of chip
asynchronism are analyzed assuming bandlimited chip pulses. In
\cite{mantravadi:02}, the chip waveform is assumed to be an ideal
Nyquist sinc function, i.e.,\ a sinc function with bandwidth equal
to half of the chip rate. The received signal is filtered
by a lowpass filter (or, equivalently, a filter matched to the
chip waveform) and subsequently sampled at the time delay of the
signal of the user of interest with a frequency equal to the chip
rate\footnote{The chip rate satisfies the condition of the
sampling theorem in this case.}. Reference \cite{mantravadi:02} proves that
the SINR at the output of the linear MMSE detector converges in
the mean-square sense to the SINR in an equivalent synchronous
system. In \cite{schramm:99} the wider class of chip pulses
which are inter-chip interference free at the output of the chip matched
filter is considered. In the following we will refer to this class
of chip pulses as square root Nyquist chip pulses.

In \cite{madhow:94,madhow:99}, the performance of the linear MMSE
detector with completely asynchronous users and chip waveforms
limited to a chip interval is analyzed. However, the observation
window in \cite{madhow:94} spanned only a single symbol interval
not yielding sufficient discrete-time statistics; the resulting
degradation in performance was pointed out later in
\cite{schramm:99,kiran:00}.

As discussed above, previous approaches to the analysis of
asynchronous CDMA with multiuser detection were only concerned
with, if and how asynchronism can be prevented from causing
performance degradation. However, asynchronism is known to be
beneficial for CDMA systems with demodulation by single-user
matched filters (e.g., \cite{pursley:77a}). One of the main
contributions of this paper is to show that benefits from asynchronism
are not inherent to single-user matched filters but they are a general
property of CDMA systems. We quantify those benefits in terms of spectral efficiency and SINRs in the large-system limit.

Compared to synchronous systems, the analysis of asynchronous CDMA
raises two additional issues: (i) the way statistics are formed,
trading complexity against performance, and (ii) the effects of
excess bandwidth, chip-pulse shaping and the users' delay
distribution.

The optimum multiuser detector in \cite{verdu:86} is based on the
sufficient statistics obtained as output samples of a bank of
filters matched to the symbol spreading waveforms of all users.
The decorrelating detector in \cite{lupas:90} and the linear MMSE
detector in \cite{rupf:94} benefit from  the same sufficient
statistics. A method to determine the eigenvalue moments of a
correlation matrix in asynchronous CDMA systems using such
statistics and square root Nyquist chip pulse waveforms is
proposed in \cite{hwang:07}.

An alternative approach to generate useful statistics, which in
general are not sufficient, is borrowed from synchronous systems.
The received signals are processed by a filter matched to the chip
waveform and sampled at the chip rate. This approach is optimum
for single-user communications and chip-synchronous multi-user
communications, but causes aliasing to the signals of
de-synchronized users if the chip waveform has non-zero excess
bandwidth.
Discretization schemes using chip matched filters and sampling at
the Nyquist rate are studied in \cite{mantravadi:01}. There, the
notion of approximate sufficient statistics was introduced.
Furthermore, conventional CDMA systems with chip waveforms that approximate
sinc pulses were shown to outperform systems using rectangular pulse shaping.
Moreover, it was conjectured that sinc pulses are optimal for CDMA systems with linear MMSE multiuser detection.

In systems with bandlimited
waveforms, sampling at a rate faster than the Nyquist rate leads
to the same performance as the optimal time-discretization proposed in
\cite{verdu:86,lupas:90,rupf:94} if the condition of the sampling
theorem is satisfied \cite{schramm:99}. In contrast to the bank of symbol matched filters in
\cite{verdu:86}, this approach has the advantage that the
time delays of the users' signals need not be known
before sampling.

The impact of the shape  and  excess bandwidth of the chip pulses
received attention in
\cite{landolsi:95,landolsi:99,mantravadi:01}. In
\cite{landolsi:95,landolsi:99} an algorithm for the design of
chip-pulse waveforms for CDMA systems with \emph{conventional
detection}  has been proposed. The design criterion consists of
minimizing the bit error rate at the output of a \emph{single user
matched filter} in asynchronous CDMA systems while enforcing certain constraints on the
chip waveforms.

This work is organized in seven additional sections. Section
\ref{sec:mainresults}, gives a brief overview of the main results
found in this work.
Sections \ref{sec:notation} and \ref{sec:system_model} introduce notation
and the system model for asynchronous CDMA, respectively. Section
\ref{sec:MMSE_det} focuses on the analysis of linear MMSE
detectors and introduces the main mathematical tools for
analysis of the fundamental limits of asynchronous CDMA.
In Section \ref{sec_capacity}, the spectral efficiency of optimal joint decoding is derived on the basis of the results for the linear MMSE detector exploiting the duality between mutual information and MMSE.
Section~\ref{conj} addresses the extension of the presented results to more general settings.
Some
conclusions are drawn in Section \ref{chap:async_sec:conclusions}.

\section{Main Results}
\label{sec:mainresults}

Before going into the main results of this work, it is helpful to
get some intuition on asynchronous CDMA systems. First of all, one
might be interested in the question which chip waveform gives the
highest spectral efficiency for otherwise arbitrary system
parameters, like pulse shape, pulse width in time and frequency,
system load, etc. There is a surprisingly easy answer to this
question that does not require any sophisticated mathematical
tools:
\begin{proposition}
Without constraints on the system parameters, Nyquist sinc-pulses maximize spectral efficiency .
\end{proposition}
\begin{proof}
The proof is by contradiction. First, it is well-known that the
spectral efficiency of a single user channel is maximized by
Nyquist sinc-pulses. Further, we know from \cite{verdu:99a} that
the spectral efficiency of a synchronous CDMA system with Nyquist
sinc-pulses becomes identical to the spectral efficiency of a
single user channel, as the load converges to infinity. Finally,
the multiuser system can never outperform the single user system,
since we could otherwise improve a single user system by virtually
splitting the single user into many virtual users. Thus, the
Nyquist sinc pulse is optimum also for the multi-user system.
\end{proof}
Note that, from the previous proof, the Nyquist sinc pulse is
optimum for an infinite system load. However, we cannot judge
whether the optimum is unique from the line of thought proposed in
our proof. In fact, a straightforward application of a more general
result in this paper (shown in the Appendix
\ref{section:proof_proposition_2}) is the following:
\begin{proposition}
\label{prop2} Asynchronous CDMA systems with any sinc-pulses, no
matter whether they are constrained to the Nyquist bandwidth or to
a larger, or even to a smaller bandwidth\footnote{It is worth noticing that any chip pulse with bandwidth smaller than the Nyquist bandwidth necessarily introduces inter-chip interference.}, and users whose
empirical delays are uniformly distributed within a symbol interval achieve the same spectral efficiency as a single
user channel, if the load converges to infinity.
\end{proposition}

The optimization of the system load neither gives the
theoretically most interesting cases to consider nor the
practically most relevant. Let us, thus, look at which chip
waveforms achieve the highest spectral efficiency for a fixed
load. Surprisingly, the result is not a little bit more useful for
practical applications:
\begin{corollary}
The chip waveform that maximizes the spectral efficiency for a
given finite load and given chip rate has vanishing bandwidth.
Furthermore, the maximum spectral efficiency is the same as the
one for the single user channel.
\end{corollary}
\begin{proof}
The corollary follows directly from Proposition~\ref{prop2}. Note that Proposition~\ref{prop2} holds for an arbitrary chip rate and an arbitrary bandwidth of the chip waveform and states that the single user bound is reached at infinite load.
Though, the corollary is stated for a given finite load, we are free to decompose each physical user into $M$ virtual users and let $M$ increase to infinity such that the virtual load becomes infinite. Therefore, we take a user's signal and divide it into $M$ data streams that are time-multiplexed in such a way as to result in the same physical transmit signal for that user.
An example of such a decomposition is illustrated in Fig.~\ref{coro1}.
\begin{figure}
\centerline{\epsfig{file=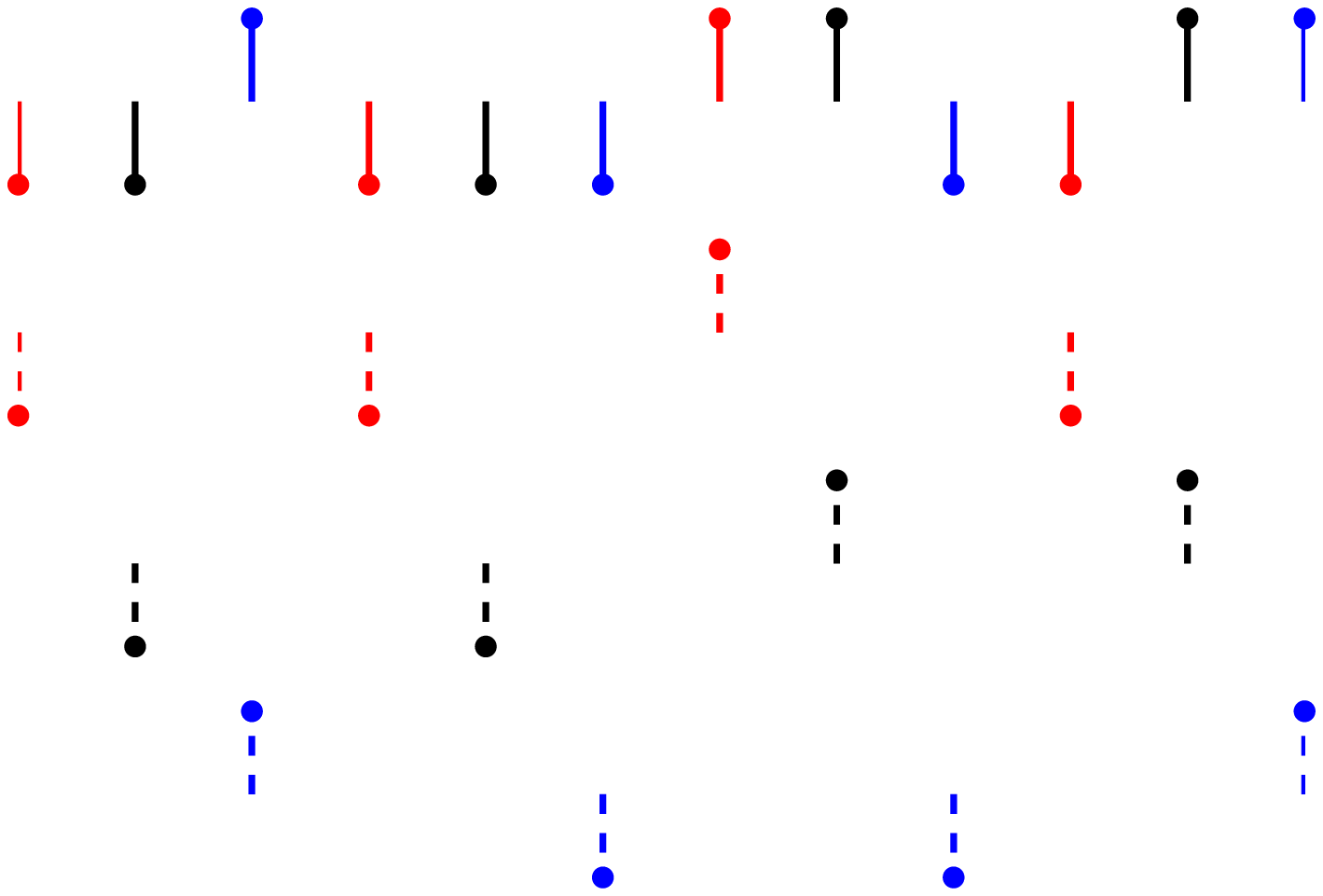, width=8.5cm,height=4.3cm}}
\caption{\label{coro1} Decomposition of a spreading sequences of length $N=12$ into $M=3$ virtual spreading sequences.}
\end{figure}
Applying this idea to each of the $K$ physical users, we have created $MK$ virtual users. Furthermore, the chip interval has grown from $T_c$ to $MT_c$ and the virtual users are asynchronous with a discrete uniform distribution of delays within the virtual chip interval of length $MT_c$. Consider now a sinc pulse of bandwidth $1/(2MT_c)$ as chip waveform. If we take the limit $M\to\infty$ for the system of virtual users, the delay distribution converges to the uniform distribution within the virtual chip interval and the number of virtual users converges to infinity. Thus, Proposition~\ref{prop2} applies and the single user bound is reached. Therefore, this choice of chip waveforms whose bandwidth vanishes is optimal.
\end{proof}

Optimizing chip waveforms to maximize spectral efficiency has proven to hardly aid the practical design of CDMA systems, since the optima are achieved for system parameters, e.g.,\ infinite load and/or vanishing bandwidth, that are far from the limits of practical implementation.
Furthermore, the choice of the chip waveform is influenced by many other factors than spectral efficiency like the difficulty to implement steeply decaying frequency filters and the need to keep the peak-to-average power ratio of the continuous-time transmit signal moderate. Therefore, many commercial CDMA systems, e.g.,\ the Universal Mobile Telecommunication System (UMTS), use chip waveforms with excess bandwidth. The UMTS standard uses root-raised cosine pulses with roll-off factor 0.22.
Motivated by the theoretical findings above and the practical constraints on the design of chip waveforms, the rest of this paper puts the focus on the performance analysis of CDMA system with a given fixed chip waveform. As it will be seen, this gives rise to a rich collection of insights into CDMA systems with asynchronous users. The main results are summarized in the following.

CDMA systems using chip pulse waveforms with bandwidth $B$ not
greater than half of the chip rate $\frac{1}{T_c},$ i.e.,\ $B \leq
\frac{1}{2 T_c},$ perform identically irrespective of whether the
users are synchronized or not for a large class of performance
measures.  Furthermore, our result generalizes the
equivalence result for the ideal Nyquist sinc waveform (with
bandwidth $\frac{1}{2 T_c}$) in \cite{mantravadi:02} to any chip
pulse satisfying the mentioned bandwidth constraint and to any
linear multistage detector and the optimal capacity-achieving
joint decoder. Note that the performance is independent of the
time delay distribution. Increasing the bandwidth of the chip
waveform above $\frac{1}{2 T_c}$, i.e.,\ allowing for some excess
bandwidth as it is customary in all implemented systems, the
behaviors of CDMA systems change substantially. They depend on the
time delay distribution and the equivalence between synchronous
and asynchronous systems is lost.

For any choice of chip waveform, we capture the performance of a
large CDMA system with linear MMSE detection by a positive
definite frequency-dependent Hermitian matrix $\boldsymbol{\Upsilon}(\Omega)$
whose size is the ratio of sampling rate to chip rate
(the sampling rate is a multiple of the chip rate). We require
neither the absence of inter-chip interference, nor that the
samples provide sufficient statistics, nor a certain delay
distribution. Unlike for synchronous users, the multiuser
efficiency \cite{verdu:98} in the large-system limit is not
necessarily unique for all users. The matrix $\boldsymbol{\Upsilon}(\Omega)$
reduces to a scalar frequency dependent function $\eta(\omega)$ in cases where oversampling is
not needed. Interestingly, the same holds true even in cases with
excess bandwidth if the delay distribution is uniform. The scalar
$\eta(\omega)$ can be understood as a multiuser efficiency
spectral density with the multiuser efficiency being its integral
over frequency $\omega$. We find that in large systems, the
effects of interference from different users and interference at
different frequencies decouple. We, thus, generalize Tse and
Hanly's \cite{tse:99b} concept of effective interference to the
concept of effective interference spectral density which decouples
the effects of interference in both user and frequency domain.

Excess bandwidth can be utilized if users are asynchronous.
While excess bandwidth is useless for synchronized systems in terms of multiuser efficiency, i.e.,\ all square root Nyquist pulses perform the same regardless of their bandwidth, desynchronizing users improves the performance of any system with non-vanishing excess bandwidth.

\section{Notation and Some Useful Definitions}\label{sec:notation}

Throughout this work, upper and lower boldface symbols are respectively used
for matrices and vectors spanning a single symbol interval.
 Matrices and
vectors describing signals spanning more than a symbol interval
are denoted by upper boldface calligraphic letters.

In the following, we utilize \emph{unitary} Fourier transforms
both in the continuous time and in the discrete time domain.
The unitary Fourier transform of a signal
$s(t)$ in the continuous time domain is given by $S( \omega)=\frac{1}{\sqrt{2 \pi}} \int_{-\infty}^{+\infty} s(t) \mathrm{e}^{-j \omega t}
\mathrm{d}t $. The unitary Fourier transform of a sequence $\{
\ldots, c_{-1}, c_{0}, c_{1}, \ldots \}$ in the discrete time
domain is given by $c(\Omega)= \frac{1}{\sqrt{2
\pi}} \sum_{n=-\infty}^{+\infty} c_n \mathrm{e}^{-j \Omega n}$.
We will refer to them shortly as
Fourier transform. Throughout this work, $\omega$ and $\Omega$ denote the angular frequency and the  angular frequency normalized to the chip rate, respectively. A function in $\Omega$ has support in the interval $[-\pi, \pi],$ or translations of it.

For further studies it is convenient to define  the concept of
\emph{$r$-block-wise circulant matrices of order $N$}:
\begin{definition} Let $r$ and $\N$ be  positive integers. An
$r$-block-wise circulant matrix of order $\N$  is an $r \N \times
\N$ matrix of the form
\begin{equation}
\Mat{C}=\left( \begin{array}{cccc}
 \Mat{B}_0 & \Mat{B}_1 & \cdots & \Mat{B}_{\N-1} \\
  \Mat{B}_{\N-1} & \Mat{B}_0 & \cdots & \Mat{B}_{\N-2} \\
  \vdots & \vdots &  &\vdots \\
 \Mat{B}_{1} & \Mat{B}_{2} & \cdots & \Mat{B}_{0} \\
\end{array} \right)
\end{equation}
with $\Mat{B}_i=(c_{1,i},c_{2,i}, \ldots, c_{r,i})^T.$
\end{definition}
In the matrix $\Mat{C}$, an $r \times N$ block row is obtained by a
circular right shift of the previous block.
Since the matrix
$\Mat{C}$ is univocally defined by the unitary Fourier transforms
of the sequences
\begin{equation} c_s(\Omega)=\frac{1}{\sqrt{2 \pi}} \sum_{k=0}^{\N-1}
c_{s,k} \mathrm{e}^{-j \Omega k } \qquad s=1, \ldots, r,
\end{equation}
there exists a bijection $\mathfrak F$ from  the frequency dependent vector $\Mat c(\Omega)=[c_1(\Omega), c_2(\Omega),\dots,c_r(\Omega)]$ to $\Mat C$.
Thus,
\begin{equation}
\mathfrak F\{\Mat c(\Omega)\} = \Mat C.
\end{equation}

Furthermore, the superscripts $\cdot^T$ and $\cdot^H$ denote the
transpose and the conjugate transpose of the matrix argument,
respectively. $\I_{n}$ is the identity matrix of size $n \times n$
and $\mathbb{C}$, $\mathbb{Z}$, $\mathbb{Z}^{+},\mathbb N,$  and
$\mathbb{R}$ are the fields of complex, integer, nonnegative
integer, positive integer, and real numbers, respectively.
$\mathrm{tr}(\cdot)$, $\| \cdot \|$, and $|\cdot|$ are the trace,
the Frobenius norm, and the spectral norm of the argument,
respectively, i.e.,\ $\|\mathbf{A}
\|=\sqrt{\mathrm{tr}(\mathbf{AA}^H)}$,
$|\mathbf{A}|=\max\limits_{\mathbf{x}^H\mathbf{x} \leq
1}\mathbf{x}^H \mathbf{A}\mathbf{A}^H \mathbf{x} $.
$\mathrm{diag}(\cdot): \mathbb{C}^{n} \mapsto \mathbb{C}^{n \times
n}$ transforms an $n$-dimensional vector into a diagonal matrix of
size $n\times n$ having as diagonal elements the components of the
vector in the same order. $\E\{\cdot\}$ and $\mathrm{Pr}\{\cdot\}$
are the expectation and probability operators, respectively.
$\delta_{ij}$ is the Kronecker symbol and $\delta(\lambda)$ is
Dirac's delta function. $\Mat{X}=(x_{ij})_{i=1,\ldots,
n_1}^{j=1,\ldots, n_2}$ is the $n_1 \times n_2$ matrix whose
$(i,j)$-element is the scalar $x_{ij}$.
$\Mat{X}=(\Mat{X}_{ij})_{i=1,\ldots, n_1}^{j=1,\ldots, n_2}$ is
the $n_1 q_1 \times n_2 q_2$ block matrix whose $(i,j)$-block is
the $q_1 \times q_2$ matrix  $\Mat{X}_{ij}$. The notation $\lfloor
\cdot \rfloor$ is adopted for the operator that yields the maximum
integer not greater than its argument and $x \mathrm{mod} y$
denotes the modulus, i.e., $x \mathrm{mod} y=x-\left \lfloor
\frac{x}{y}\right\rfloor y$. Furthermore, $\chi(\Vec{x} \in
\mathcal{A})$ denotes the indicator function of the variable
$\Vec{x}$ on the set $\mathcal{A}$ and $\chi(\Vec{x} \in
\mathcal{A})=1$ if $\Vec{x} \in \mathcal{A}$ and zero otherwise.

\section{System Model}\label{sec:system_model}

Let us consider an asynchronous CDMA system with $K$ users in the
uplink channel. Each user and the base station are equipped with a
single antenna. The channel is flat\footnote{
Flat fading is no restriction of generality here as long as the excess delay is much smaller than the symbol interval $T_s$. This is, as the effect of multi-path can be incorporated into the shape of the chip waveform.
} fading and impaired by
additive white Gaussian noise. Then, the signal received at the
base station, in complex base-band notation, is given by
\begin{equation}\label{general_continuous_system}
   y(t)=\sum_{k=1}^{K} a_{k} s_k(t-\tau_k) + w(t)
   \qquad t \in (-\infty, +\infty).
\end{equation}
Here, $a_{k}$ is the received signal amplitude of  user $k$, which
takes into account the transmitted amplitude, the effects of the
flat fading, and the carrier phase offset; $\tau_k$ is the time
delay of user $k$; $w(t)$ is a zero mean white, complex Gaussian
process with power spectral density $N_0$; and $s_k(t)$
is the spread signal of user $k$. We have
\begin{equation}
   s_k(t)=\sum_{m=-\infty}^{+\infty} b_k[m] c_k^{(m)}(t),
\end{equation}
where $b_k[m]$ is the $m^{\mathrm{th}}$ transmitted symbol of user
$k$ and
\begin{equation}\label{spreading_waveform}
c_k^{(m)}(t)=\sum_{n=0}^{N-1} s_{k,m}[n] \psi(t-mT_s-nT_c)
\end{equation}
is its spreading waveform at time $m$. Here, $\vecs_{k}^{(m)}$ is the
spreading sequence vector of user $k$ in the $m^{\mathrm{th}}$ symbol
interval 
 with elements $s_{k,m}[n]$, $n=0,\ldots, \N-1$. $T_s$ and
$T_c=\frac{T_s}{\N}$ are the symbol and chip interval,
respectively.

The users' symbols $b_k[m]$ are independent and identically
distributed (i.i.d.) random variables with $\mathrm{\mathrm{E}}\{|b_k[m]|^2
\}=1$ and $\mathrm{E}\{b_k[m]\}=0$.  The elements of the spreading
sequences $s_{k}^{(m)}[n]$ are assumed to be i.i.d.\ random variables
with $\mathrm{E}\{|s_{k,m}[n]|^2\}=\frac{1}{N}$ and
\mbox{$\mathrm{E}\{s_{k,m}[n]\}=0$}. This assumption properly
models the spreading sequences of some CDMA systems currently in
use, such as the long spreading codes of the FDD (Frequency
Division Duplex) mode in the UMTS uplink channel.

The chip waveform $\psi(t)$ is limited to bandwidth $B$
and energy $E_{\psi}=\int_{-\infty}^{+\infty}|\psi(t)|^2
\mathrm{d}t.$ Because of the  constraint on the variance of the
chips, i.e.,\ $\E\{|s_{k,m}[k]|^2 \}= \frac{1}{\N}$, the mean
energy of the signature waveform satisfies $\E \left \{
\int_{-\infty}^{+\infty}|c_k^{(m)}(t)|^2 \mathrm{d}t \right
\}=E_{\psi} $. We assume \begin{enumerate}
\item user 1 as reference user so that $\tau_1=0$,
\item the users are ordered according to increasing time delay with respect to the reference
user,
\item the time delay to be, at most, one chip interval so
that $\tau_k \in [0,T_c)$.
\end{enumerate}
Assumptions 1 and 2 are without loss of generality
\cite{verdu:93}. Assumption 3 is  made for the sake of
clarity and it will be removed in Section \ref{conj} where the
results are extended to the general case with $\tau_k \in
[0,T_s).$

At the receiver front-end, the base band signal is passed through
a filter with impulse response $g(t)$ and corresponding transfer
function $G(\omega)$ normalized such that
$\int_{-\infty}^{+\infty}|g(t)|^2 {\rm d}t =1$. We denote by
$\phi(t)$ the response of the filter to the input $\psi(t),$ i.e.,\
$\phi(t)= g(t) \ast \psi(t)$ and by $\Phi(\omega)$ its Fourier
transform. The filter output is sampled at rate $\frac{r}{T_c}$
with $r\in\mathbb N$. For further convenience, we also define
$E_\phi=\int_{-\infty}^{+\infty}|\phi(t)|^2{\rm d}t$.

Throughout this work we assume that the filtered chip pulse
waveform ${\phi}(t)$ is much shorter than the symbol waveform,
i.e.,\ ${\phi}(t)$ becomes negligible for $|t| > t_0$ and $t_0 \ll
T_s$. This technical assumption is usually verified in the systems
with large spreading factor we are considering. It allows to
neglect intersymbol interference. Thus, focusing on a given symbol
interval, we can omit the symbol index $m$ and the discrete-time
signal at the front-end output is given by
\begin{align}\label{discrete_output}
   y[p]=\sum_{k=1}^{K}a_{k}  b_k
   \overline{c}_{k}\left( \frac{p}{r}T_c - \tau_k \right)+w[p]
\end{align}
with sampling time $p \in \{0,\dots, rN{-}1\}$ and
\begin{equation}
\overline{c}_{k}(t)=\sum_{n=0}^{N-1} s_{k}[n] {\phi}
\left(t-n T_c \right).
\end{equation}
Here, $w[p]$ is discrete-time, complex-valued noise. In general,
$w[p]$ is not white, although the continuous process was white. However, it {\em is} white, if $g(t)\ast g(-t)$ is Nyquist with respect to the sampling rate. In this latter case, the noise variance is $\sigma^2=\frac{N_0 r}{T_c}.$ This expression accounts for the normalization of the front end filter.

In order to cope with the effects of oversampling, we consider an extended signal space with virtual spreading sequences of length $rN$.
The virtual spreading sequence of user $k$ is given by the $N r$-dimensional vector
\begin{equation}
\Mat{\overline\Phi}_{k}
\Vec{s}_{k}
\end{equation}
where
$\Vec{s}_{k} = (s_{k}[0] \ldots s_{k}[N-1])^{T},$
\begin{equation}
\label{block_wise_toeplitz_matrix}
   \overline{\Mat{\Phi}}_k= \left( \begin{array}{cccc}
     \Vec{\phi}_{-\tau_k} & \Vec{\phi}_{- T_c \!-\! \tau_k} & \ldots & \Vec{\phi}_{(-\N\!+\!1)T_c\!-\! \tau_k} \\
                 \Vec{\phi}_{T_c\!-\!\tau_k} & \Vec{\phi}_{-\tau_k} & \ldots & \Vec{\phi}_{(-\! \N\!+\!2)T_c\!-\!\tau_k} \\
 \vdots & \vdots & \vdots & \vdots \\       \Vec{\phi}_{(\!\N\!-\!1\!)T_c\!-\!\tau_k} & \Vec{\phi}_{(\N\!-\!2 )T_c\!-\!\tau_k} & \cdots & \Vec{\phi}_{\!-\!\tau_k} \\
\end{array} \right)
\end{equation}
is a $Nr \times N$ block matrix taking into account
the effects of delay and pulse shaping. Its blocks are the vectors $\Vec{\phi}_x=\left(\phi(x), \phi\left(x+\frac{T_c}{r}\right), \ldots, \phi\left(x+\frac{r-1}{r}T_c\right)\right)^T.$
In that way, we have described user $k$'s continuous-time channel with continuous delays canonically by the discrete-time channel matrix $\Mat{\overline\Phi}_k$.
Note that $\Mat{\overline\Phi}_k$ solely depends on the delay of user $k$, the oversampling factor $r$, the chip waveform, and the receive filter.

Structuring the matrix $\overline{\Mat{\Phi}}_k$ in blocks of
dimensions $r \times 1,$ it is block-wise
Toeplitz.  As well known \cite{gray:71,gray:72},
block-Toeplitz and block-circulant matrices are
asymptotically equivalent in terms of spectral distribution.
This asymptotic equivalence is sufficient for us, since  in this work we focus on performance measures of CDMA systems which depend only on the asymptotic eigenvalue distribution.
Similar asymptotically tight approximations
are used in the large system analysis of CDMA in frequency-selective
fading \cite{evans:00,li:04,cottatellucci:04}.

The equivalent block-circulant matrix is given by
\begin{equation}\label{matrice_Phi}
\Mat{\Phi}_{k} = \mathfrak F\left\{\left[{\phi}(\Omega,\tau_k),
{\phi}\left(\Omega,\tau_k-\tfrac{T_c}{r}\right),\ldots,
{\phi}\left(\Omega,\tau_k-\tfrac{(r-1)T_c}{r}\right)\right]\right\},
\end{equation}
where
\begin{equation}\label{discrete_Fourier_phi} {\phi}(\Omega,\tau)\overset{\triangle}{=} \frac{1}{T_c}
\sum_{\nu=-\infty}^{+\infty} \mathrm{e}^{  j \frac{
\tau}{T_c}(\Omega+ 2\pi \nu)} \Phi^{*} \left(\tfrac{1 }{T_c}(\Omega+ 2\pi \nu) \right)
\end{equation}
is the spectrum of the chip waveform $\phi(t)$ delayed by $\tau$ and sampled at rate $1/T_c$.
Thus, we replace the block-Toeplitz matrix $\Mat{\overline\Phi}_k$ for our asymptotic analysis by the block-circulant matrix $\Mat{\Phi}_{k}$ in the following and use the virtual spreading sequences
\begin{equation}
\Vec v_k =\Mat\Phi_k \Vec s_k.
\end{equation}

Let $\Mat{S}$ be the $rN \times K$ matrix of virtual
spreading, i.e.,\ $\Mat{S}=(\Mat{\Phi}_1 \Vec{s}_{1},
\Mat{\Phi}_2 \Vec{s}_{2}, \ldots \Mat{\Phi}_K \Vec{s}_{K}),$
$\Mat{A}$ the $K \times K$ diagonal matrix of received amplitudes,
$\Mat{H}=\Mat{S} \Mat{A},$ and $\Vec{b}$ and
\begin{equation}
\label{matrix_model}
\Vec y =\Mat{Hb}+\Mat w
\end{equation}
the vectors of transmitted and received signals, respectively.
Additionally, $\Vec{h}_{k}$ denotes the $k^{\mathrm{th}}$ column of the matrix
$\Mat{H}.$ Finally, we define the correlation
matrices $\Mat T=\Mat H\Mat H^H$,
$\Mat R=\Mat H^H \Mat H$ and the system load $\beta=\frac{K}{N}$.


\section{Linear MMSE Detection}\label{sec:MMSE_det}

The linear MMSE detector $\Vec{{d}}_{k}$ generates a
soft decision $\widehat{b}_k=\Vec{{d}}_{k}^H \Mat y$
of the transmitted symbol $b_k$ based on the observation
$\Mat y$.
It can be derived from the Wiener-Hopf theorem \cite{Kaybook}
and is given by \be \Vec{{d}}_{k}=\E\{\Mat y
\Mat y^H\}^{-1} \E\{b_k^{*} \Mat y\} \ee with the expectation
taken over the transmitted symbols $\Mat b$ and the noise.
Specializing the Wiener-Hopf equation to the system model
(\ref{matrix_model}) yields
\begin{eqnarray}\label{MMSE_detector}
\Vec{d}_{k}&=&(\Mat H \Mat H^H+\sigma^2 \Mat{I})^{-1}\Vec{h}_{k}  \\
&=& c \cdot (\Mat H_{k} \Mat H_{k}^H+\sigma^2
\Mat{I})^{-1}\Vec{h}_{k}
\end{eqnarray}
for some $c\in\mathbb{R}$. Here, $\Mat H_{k}$ is the matrix
obtained from $\Mat H$ suppressing column $\Vec{h}_{k}$.
The second step follows from the matrix inversion lemma.

The performance of the linear MMSE detector is measured by the
signal-to-interference-and-noise ratio at its output
\cite{verdu:98} \be\label{SINR}
\mathrm{SINR}_{k}=\Vec{h}_{k}^H(\Mat H_{k}\Mat H_{k}^H+\sigma^2
\Mat{I})^{-1}\Vec{h}_{k}. \ee The SINR can be conveniently
expressed in terms of the multiuser efficiency $\eta_{k}$
\cite{verdu:98} \be \label{defMUE} \mathrm{SINR}_{k} =
\frac{|a_{k}|^2 E_{\phi}}{N_0}\, \eta_{k}. \ee
The multiuser
efficiency is a useful measure, since for large systems it is
identical for all users in special cases \cite{shitz:99} and it is
related to the spectral efficiency \cite{shitz:99,mueller:03,guo:04}.

The SINR depends on the spreading
sequences, the received powers of all users, the chip pulse
shaping, and the time delays of all users.
To get deeper insight on the linear MMSE detector it is
convenient to analyze the performance for random spreading sequences in the large system limits, i.e.,\ as $K, N \rightarrow \infty
$ with the ratio $\frac{K}{N} \rightarrow \beta$ kept fixed.
The large-system analysis will identify the
macroscopic parameters that characterize a chip-asynchronous CDMA
system and the influences of chip pulse shaping and delay
distribution.

In this section we present  the large system analysis of a linear
MMSE detector for chip-asynchronous CDMA systems with random
spreading. Provided that the noise at the output of the front end
is white, the analysis applies to CDMA systems using either
optimum or suboptimum statistics, any chip pulse waveforms, and
any set of time delays in $[0,T_c)$ if their empirical
distribution function converges to a deterministic limit.

In Appendix \ref{section:proof_theo_sinr_MMSE_chip_asynch}, we derive the following theorem on the large-system performance of chip-asynchronous CDMA:

\begin{theor}\label{theo:sinr_MMSE_chip_asynch}
Let $\matA \in \mathbb{C}^{\K \times \K}$ be a diagonal matrix
with $k^{\text{th}}$ diagonal element $a_{k}\in \mathbb{C}$ and $T_c$ a positive
real. Given a function $\Phi(\omega): \mathbb{R} \rightarrow
\mathbb{C}$, let $\phi(\Omega,\tau)$ be as in
(\ref{discrete_Fourier_phi}). Given a positive integer $r$, let
$\Mat\Phi_k$, $k=1, \ldots, \K$, be $r$-block-wise circulant
matrices of order $\N$ defined in (\ref{matrice_Phi}).
 Let
${\matH}={\matS}\matA$ with
${\matS}=[\Mat\Phi_1\vecs_{1},
\Mat\Phi_2 \vecs_{2}, \ldots,
\Mat\Phi_K\vecs_{\K}]$ and $\vecs_k \in\mathbb{C}^{N\times 1}$.

Assume that the function $|\Phi(\omega)|$ is upper bounded and has
finite support. The receive filter is such that the sampled
discrete-time noise process is white. The vectors $\vecs_{k}$ are
independent with i.i.d.\ circularly symmetric Gaussian  elements.
Furthermore, the elements $a_{k}$ of the matrix $\matA$ are
uniformly bounded for any $K.$  The sequence of the empirical
joint distributions
$F_{|\matA|^2,T}^{(\K)}(\lambda,\tau)=\frac{1}{\K} \sum_{k=1}^{\K}
\chi(\lambda>|a_{k}|^2)\chi(\tau>{\tau}_k)$ converges almost
surely, as $\K \rightarrow \infty$, to a non-random distribution
function $F_{|\matA|^2,T}(\lambda,\tau).$

Then, given the received power $|a_{k}|^2$, the time delay
${\tau}_k,$ and the variance of the white noise
$\sigma^2=\frac{rN_0}{T_c}$, the SINR of user $k$ at the output of
a linear MMSE detector for a CDMA system with transfer matrix
${\matH}$ converges in probability as $K, N \rightarrow \infty$
with $\frac{K}{N} \rightarrow \beta$ and $r$ fixed to

\begin{equation}\label{limit_SINR_MMSE_genaral}
\lim_{K=\beta N \rightarrow \infty} \mathrm{SINR}_k =
\frac{{|a_{k}|^2}}{2\pi} \int\limits_{-\pi}^{\pi}
\Mat{\Delta}_{\phi,r}^H(\Omega, {\tau}_k) \Mat{\Upsilon}(\Omega)
\Mat{\Delta}_{\phi,r}(\Omega, {\tau}_k) \mathrm{d}\Omega
\end{equation}
where $\Mat{\Upsilon}(\Omega)$ is the unique positive definite $r \times
r $ matrix solution of the fixed point matrix equation\footnote{Here, the integration measure is meant to denote ${\rm d}F_{|\Mat A|^2,T}(\lambda,\tau) = f_{|\Mat A|^2,T}(\lambda,\tau) {\rm d}\lambda{\rm d}\tau$ in case such a representation exists.}
\begin{equation}\label{multiuser_efficiency_mat}
\Mat{\Upsilon}^{-1}(\Omega) = \sigma^2 \Mat{I}_{r} + \beta
\int \dfrac{\lambda
\Mat{\Delta}_{\phi,r}(\Omega, \tau)\Mat{\Delta}_{\phi,r}^H(\Omega, \tau)
{\rm d} F_{|\matA|^2, T } (\lambda, \tau)} {1 + \frac\lambda{2\pi}
\int_{-\pi}^{\pi} \Mat{\Delta}_{\phi,r}^H(\Omega^{\prime},
\tau)\Mat{\Upsilon}(\Omega^{\prime})\Mat{\Delta}_{\phi,r}(\Omega^{\prime}, \tau) \mathrm{d}\Omega^{\prime} }
\qquad -\pi < \Omega \leq \pi,
\end{equation}
and
\begin{equation}\label{delta_phi_r_def}\Mat{\Delta}_{\phi,r}(\Omega,{\tau})=\left(
\begin{array}{c}
 {\phi}(\Omega, \tau) \\
 {\phi}(\Omega,\tau-\frac{T_c}{r}) \\
 \vdots \\
 {\phi}(\Omega, \tau-\frac{T_c (r-1)}{r}) \\
\end{array}
\right).
\end{equation}
\end{theor}

The performance of the linear MMSE detector operating on not necessarily sufficient statistics is
completely characterized by
\begin{enumerate}
\item
 an $r \times r$ matrix-valued transfer function
$\Mat{\Upsilon}(\Omega)$ and
\item the frequency and delay dependent vector $\Mat{\Delta}_{\phi,r}(\Omega,
\tau)$
\end{enumerate}
The multiuser efficiency varies from user to user and depends on
the time delay of the user of interest only through
$\Mat{\Delta}_{\phi,r}(\Omega, \tau_k).$ We can
define an SINR spectrum
\begin{equation}
\mathrm{SINR}_{k}(\Omega)= \frac{|a_{k}|^{2}}{2\pi}
\Mat{\Delta}_{\phi,r}^H(\Omega, \tau_k) \Mat{\Upsilon}(\Omega)
\Mat{\Delta}_{\phi,r}(\Omega, \tau_k)
\end{equation}
in the normalized frequency domain $-\pi < \Omega \leq
\pi,$ or, equivalently, a spectrum of the multiuser
efficiency. The system performance is in both cases obtained by integration over the
spectral components.

The fixed point equation \eq{multiuser_efficiency_mat} clearly reveals how and why synchronous users are the worst case for a given chip waveform.
We know from \cite{guo:04} that to each large multiuser system, there is an equivalent single user system with enhanced noise, but otherwise identical performance.
In the present case with oversampling factor $r$, the equivalent single user system is a frequency-selective MIMO (multiple-input multiple-output) system with $r$ transmit and $r$ receive antenna and governed by the $r\times r$ channel transfer matrix
\begin{equation}
\beta\int \dfrac{\lambda
\Mat{\Delta}_{\phi,r}(\Omega, \tau)\Mat{\Delta}_{\phi,r}^H(\Omega, \tau)
\mathrm{d} F_{|\matA|^2, T } (\lambda, \tau) }{1 + \frac\lambda{2\pi}
\int_{-\pi}^{\pi} \Mat{\Delta}_{\phi,r}^H(\Omega^{\prime},
\tau)\Mat{\Upsilon}(\Omega^{\prime})\Mat{\Delta}_{\phi,r}(\Omega^{\prime}, \tau) \mathrm{d}\Omega^{\prime} }.
\end{equation}
Note that this matrix is an integral of an outer product over the delay distribution.
Thus, for constant delay, i.e.,\ chip-synchronization, the matrix has rank one.
No additional dimensions in signal space can be spanned.
For distributed delays, the rank of the matrix can be as large as the oversampling factor $r$.
Driving the equivalence even further, the equivalent MIMO system can be transformed into an equivalent CDMA system with spreading factor $r$ and spreading sequences $\Mat{\Delta}_{\phi,r}(\Omega, \tau_k)$.
In this model, equal delays in the real CDMA system correspond to users with identical signature sequences in the equivalent CDMA system.

One cannot increase performance unboundedly by faster oversampling, as not all modes of the equivalent $r$-dimensional MIMO system can be excited with a chip waveform of limited excess bandwidth due to the projection onto the spectral support of the chip waveform in \eq{limit_SINR_MMSE_genaral}.
In order to utilize the excess bandwidth of the system, we need two ingredients: 1) Time delays separating the users by making the signatures in the equivalent system differ. 2) A receiver that transforms the continuous-time receive signal into sufficient discrete-time statistics, e.g.,\ by oversampling.
A lack of different delays leads to a system where only a single eigenmode of the equivalent MIMO system is excited.
A lack of oversampling leads to a system where more eigenmodes are excited, but are not converted into discrete time.

Additional intuitive insight into the behavior of the asynchronus CDMA
systems can be gained by focusing on CDMA systems with uniformly distributed delay. In this case, Theorem \ref{theo:sinr_MMSE_chip_asynch} can be formulated with a single scalar fixed point equation by moving from the frequency $\Omega$ that is normalized to the chip rate to the unnormalized frequency $\omega$.
This yields the following corollary:

\begin{corollary}\label{cor_MMSE_raised_cosine}
Let us adopt the same definitions as in Theorem
\ref{theo:sinr_MMSE_chip_asynch} and let the assumptions of
Theorem \ref{theo:sinr_MMSE_chip_asynch} be satisfied.
Additionally, assume that
the random variables $\lambda$ and $\tau$ in $F_{|\matA|^2,
T}(\lambda, \tau)$  are statistically independent and the random
variable $\tau$ is uniformly distributed in $[0, T_c)$.
Furthermore, let $\Phi(\omega)$ vanish outside the interval $[-2 \pi B;+ 2 \pi B]$
with $B \le \frac r{2T_c}$.
Then, the multiuser efficiency of the
linear MMSE detector for CDMA converges in probability
as $K,N \rightarrow \infty$ with $\frac{K}{N}\rightarrow \beta$
and $r$ fixed to
\begin{equation}\label{SINR_cor_MMSE_raised_cos}
\lim_{K= \beta N \rightarrow \infty} \eta_k =
\eta =\frac1{2\pi}
\int\limits_{-2\pi B}^{+2\pi B}\eta\left(\omega \right)
\mathrm{d}\omega
\end{equation}
where the multiuser efficiency spectral density
$\eta\left(\omega\right)$ is the unique solution
to the fixed point equation
\begin{equation}\label{fixed_point_eq_cor_MMSE_raised_cos}
\frac 1{\eta\left(\omega\right)} = \frac{ E_{\phi}}{ |\Phi\left(\omega\right)|^2} + \frac{\beta}{T_c} \int \frac{ \lambda \mathrm{d}F_{|\matA|^2}(\lambda)}{
\frac{N_0}{E_\phi }+ \lambda\eta}
\end{equation}
and is zero for $|\omega|>2 \pi B$.
\end{corollary}
Theorem \ref{theo:sinr_MMSE_chip_asynch} is specialized to
Corollary \ref{cor_MMSE_raised_cosine} in Appendix
\ref{section:proof_cor_MMSE_raised_cosine}.

Under the conditions of Corollary \ref{cor_MMSE_raised_cosine} the
multiuser efficiency of the linear MMSE detector in asynchronous
systems is the same for all users.

Rewriting (\ref{SINR_cor_MMSE_raised_cos}) and
(\ref{fixed_point_eq_cor_MMSE_raised_cos}) in terms of SINRs,
these equations can be interpreted similarly to the corresponding equations in
\cite{tse:99b} for synchronous systems when the concept of
effective interference is generalized to the concept of effective
interference spectral density.
Let $P(\omega, \lambda)=
\frac{\lambda }{T_c} \left|\Phi\left(\omega \right)
\right|^2$ be the power spectral density of the received
signal for a user having received power $\lambda$. Then, the
result in Corollary \ref{cor_MMSE_raised_cosine} can be expressed
as
\begin{equation}
\mathrm{SINR}_k=\frac1{2\pi} \int\limits_{-2 \pi B}^{+2 \pi B} \mathrm{sinr}_{k}(\omega) \mathrm{d}\omega
\end{equation}
where the SINR spectral density $\mathrm{sinr}_k(\omega)$ is
given by
\begin{equation}
\mathrm{sinr}_k(\omega)=\frac{P(\omega,|a_{k}|^2)}{{N_0} + \beta
\mathrm{E}_{\lambda} \{ I(P(\omega,|a_{k}|^2), P(\omega,\lambda),
\mathrm{SINR}_k) \}}
\end{equation}
with the effective interference spectral density
\begin{equation}
I(P(\omega,|a_{k}|^2), P(\omega,\lambda), \mathrm{SINR}_k)=
\frac{P(\omega,|a_{k}|^2)
P(\omega,\lambda)}{P(\omega,|a_{k}|^2)+P(\omega,\lambda)
\mathrm{SINR}_k}.
\end{equation}
Heuristically, this means that for large systems the SINR spectral density is deterministic and given by
\begin{equation}
\mathrm{sinr}_k(\omega) \approx \frac{P(\omega,|a_{k}|^2)}{N_0 +
\frac{1}{N} \sum_{ j \neq k } I(P(\omega,|a_{k}|^2),
P(\omega,|a_{j}|^2), \mathrm{SINR}_k) }.
\end{equation}
This result yields an interpretation of the effects of each of the
interfering users on the SINR of user $k$ similar to the case of
synchronous systems in \cite{tse:99b}. The impairment at frequency
$\omega$ can be decoupled into a sum of the background noise and
an interference term from each of the users at the same frequency.
The cumulated interference spectral density at frequency $\omega$
depends only on the received power density of the user of interest
at this frequency, the received power spectral density of the
interfering users at this frequency, and the attained SINR of user
$k$. In other words, in asynchronous systems we have a decoupling of the effects
of interferers like in synchronous systems and an additional
decoupling in frequency. The term $I(P(\omega,|a_{k}|^2),
P(\omega,|a_{j}|^2), \mathrm{SINR}_k)$ is the effective
interference spectral density of user $j$ onto user $k$ at
frequency $\omega$ for a given SINR of user $k$.

Sinc waveforms have a particular theoretical interest. In the following we specialize
Corollary \ref{cor_MMSE_raised_cosine} to this case.

\begin{corollary}\label{cor_sinc_pulse}
Let us adopt the definitions in Theorem
\ref{theo:sinr_MMSE_chip_asynch} and let the assumptions of
Corollary \ref{cor_MMSE_raised_cosine} be satisfied.
Given a positive real $\alpha,$  we assume that
\begin{equation}\label{sinc_waveform}
\Phi(\omega) =
 \begin{cases}
   \sqrt{\frac{T_c}{\alpha}} & {\rm for  }\quad  |\frac\omega{2\pi}| \leq \frac{\alpha}{2 T_c}, \\
   0 & {\rm otherwise}.
 \end{cases}
\end{equation}
corresponding to a sinc waveform with bandwidth
$B=\frac{\alpha}{2T_c}$ and unit energy. Then, the multiuser efficiency of the linear MMSE detector converges in probability
as $K,N \rightarrow \infty$ with $\frac{K}{N}\rightarrow \beta$
 to
\begin{equation}\label{SINR_MMSE_sinc}
\lim_{K= \beta N \rightarrow \infty} \eta_k =
\eta_{\rm sinc}
\end{equation}
where the multiuser efficiency $\eta_{\rm sinc}$ is the unique positive
solution to the fixed point equation
\begin{equation}\label{fix_point_aux}
\frac{1}{\eta_{\rm sinc}}= 1 +
\frac{\beta}{\alpha} \int \frac{\lambda
\mathrm{d}F_{|\matA|^2}(\lambda)}{N_0+\lambda
\eta_{\rm sinc}}.
\end{equation}
\end{corollary}

We recall that the multiuser efficiency of a linear MMSE detector
for a synchronous CDMA system satisfies \cite{tse:99b}
\begin{equation}\label{fixed_point_MMSE_sinc}
\frac{1}{\eta_{\mathrm{syn}} }= 1 + \beta \int
\frac{\lambda \mathrm{d}F_{|\matA|^2}(\lambda)}{N_0+\lambda
\eta_{\mathrm{syn}}}.
\end{equation}
This result holds
for synchronous CDMA systems using any chip
pulse waveform with bandwidth $B \geq \frac{1}{2 T_c}$ and
satisfying the Nyquist criterion.
Thus, it also applies to sinc pulses whose bandwidth is an integer multiple of $\frac1{2T_c}$.
Then, Corollary
\ref{cor_sinc_pulse} shows the interesting effect that an
asynchronous CDMA system using a sinc function with bandwidth
$B=\frac{\alpha}{2T_c}$ as chip pulse waveform performs as well as
a synchronous CDMA system with bandwidth $\frac r{2T_c},r\in \mathbb{N},$ and system load $\beta^{\prime}=
\frac{\beta}{\alpha}$.
This implies that only asynchronous CDMA has the capability to trade the
excess bandwidth of the chip pulse waveform against the spreading factor while synchronous CDMA has not.
In other words, asynchronous CDMA offers to trade degrees of freedom in the
frequency domain provided by the excess bandwidth of the chip pulse
waveform against degrees of freedom in the time domain
provided by spreading.

This phenomenon is
similar to the resource pooling in CDMA systems with spatial
diversity discovered in \cite{hanly:99a}. There, the degrees of
freedom in space provided by multiple antennas at the receiver
could be traded against degrees of freedom in time provided by the
spreading.
In order to make resource pooling happen, it is necessary that the steering vectors of the antenna arrays point into different directions. This condition is equivalent to requiring de-synchronization among users. If all users experience the same delay, this is like having totally correlated antenna elements.

In Corollary \ref{cor_sinc_pulse}, the bandwidth of the sinc waveform may be either larger or smaller than the Nyquist bandwidth.
For larger bandwidth, we get a resource pooling effect, for smaller bandwidth we create inter-chip interference and what could be called {\em anti-resource pooling}.
Inter-chip interference is no particular cause of concern.
In contrast, the effect of anti-resource pooling is to virtually increase the load, i.e.,\ squeezing the same number of data into a smaller spectrum is equivalent to squeezing more users into the same spectrum.
Since spectral efficiency of optimum joint decoding is an increasing function of the load \cite{verdu:99a}, anti-resource pooling is beneficial for spectral efficiency, though its implementation may cause some practical challenges.

In the following theorem, we extend anti-resource pooling to
arbitrary delay distributions:
\begin{theor}\label{theo:small_bandwidth_MMSE}
Let $\matA \in \mathbb{C}^{\K \times \K}$ be a diagonal matrix
with $k^{\text{th}}$ diagonal element $a_{k}\in \mathbb{C}$ and $T_c$ a positive
real. Given a function $\Phi(\omega): \mathbb{R} \rightarrow
\mathbb{C}$, let $\phi(\Omega,\tau)$ be as in
(\ref{discrete_Fourier_phi}). Given a positive integer $r$, let
$\Mat\Phi_k$, $k=1, \ldots, \K$, be $r$-block-wise circulant
matrices of order $\N$ defined in (\ref{matrice_Phi}).
 Let
${\matH}={\matS}\matA$ with
${\matS}=[\Mat\Phi_1\vecs_{1},
\Mat\Phi_2 \vecs_{2}, \ldots,
\Mat\Phi_K\vecs_{\K}]$ and $\vecs_k \in\mathbb{C}^{N\times 1}$.

Assume that the function $|\Phi(\omega)|$ is upper bounded and has
support  contained in the interval $
\left[-\frac{\pi}{T_c}, \frac{\pi}{T_c} \right]$. The receive
filter is such that the sampled discrete-time noise process is
white. The vectors $\vecs_{k}$ are independent with i.i.d.\
circularly symmetric Gaussian  elements. Furthermore, the elements
$a_{k}$ of the matrix $\matA$ are uniformly bounded for any $K.$
The sequence of the empirical distributions
$F_{|\matA|^2}^{(\K)}(\lambda)=\frac{1}{\K} \sum_{k=1}^{\K}
\chi(\lambda \leq |a_{k}|^2)$ converges in law almost surely, as $\K
\rightarrow \infty$, to a non-random distribution function
$F_{|\matA|^2}(\lambda).$

Then, the multiuser efficiency of the linear MMSE detector for CDMA with
transfer matrix $\matH$ converges
in probability as $K, N \rightarrow \infty$ with $\frac{K}{N}
\rightarrow \beta$ and $r$ fixed to
\begin{equation}\label{SINR_theo_small_bandwidth}
\lim_{K= \beta N \rightarrow \infty} \eta_k = \eta =\frac1{2\pi}
\int\limits_{-\pi/T_c}^{+\pi/T_c}\eta\left(\omega\right)
\mathrm{d}\omega
\end{equation}
where the multiuser efficiency spectral density $\eta(\omega)$ is
the unique solution to the fixed point equation
\begin{equation}\label{fixed_point_eq_theo_small_bandwidth}
\frac 1{\eta\left(\omega\right)} = \frac{
E_{\phi}}{|\Phi(\omega)|^2} + \frac{\beta}{T_c}
\int \frac{ \lambda \mathrm{d}F_{|\matA|^2}(\lambda)}{
\frac{N_0}{E_\phi}+ \lambda\eta}
\end{equation}
 for all $\omega$ in the support of $\Phi(\omega)$ and
zero elsewhere.
\end{theor}
Theorem \ref{theo:small_bandwidth_MMSE} is proven in Appendix
\ref{section:proof_theo_small_bandwidth_MMSE}.

No constraint is imposed on the set of time delays in Theorem \ref{theo:small_bandwidth_MMSE}.
It holds for any
set $\{ {\tau}_1, {\tau}_2 \ldots
{\tau}_K\}$ and we conclude that linear MMSE detectors
for synchronous and asynchronous CDMA systems have the same
performance if the bandwidth of the chip pulse waveforms satisfies
the constraint $B \leq \frac{1}{2 T_c}.$


\section{Spectral Efficiency}\label{sec_capacity}
There exists a close relation between the total capacity of a CDMA
system and the multiuser efficiency of a linear MMSE detector for
the same system \cite{shitz:99},
\cite{mueller:03}, \cite{guo:04}. The rationale behind this
relation is a fundamental connection between mutual information and minimum mean-squared error in Gaussian channels \cite{guo:05}.
In the following, we extend
the results in Section \ref{sec:MMSE_det} to get insight into the
spectral efficiency of an asynchronous CDMA system.

The capacity of the CDMA channel was found in \cite{verdu:86b} for
synchronous CDMA systems. The total capacity per chip  for large
synchronous CDMA systems with square root Nyquist pulses and random spreading in the presence of AWGN
(additive white Gaussian noise) is \cite{verdu:99a}

\begin{align}
\mathcal{C}^{(\text{syn})}(\beta, \mathrm{SNR})&=\beta
\log_{2}\left(1+\mathrm{SNR}-\frac{1}{4} \digamma(\mathrm{SNR},
\beta) \right) + \log_{2}\left(1+\beta \mathrm{SNR} -\frac{1}{4}
\digamma(\mathrm{SNR}, \beta)\right) \nonumber \\ & - \frac{
\log_2 e }{4 \,\mathrm{SNR}}  \,\digamma(\mathrm{SNR}, \beta)
\label{capacity_sync}
\end{align}
with
\begin{equation}
\digamma(y,z)= \left( \sqrt{y(1+\sqrt{z})^2+1}-
\sqrt{y(1-\sqrt{z})^2+1} \right)^2.
\end{equation}
With the normalizations adopted in the system model, we have
$\mathrm{SNR}=E_\phi/N_0.$

The spectral efficiency of a synchronous CDMA system is equal
to $\mathcal{C}^{(\text{syn})}(\beta, \mathrm{SNR})$ for any
Nyquist sinc waveform.
For other chip waveforms, we need to take into account the excess bandwidth and calculate spectral efficiency as
\begin{equation}
\Gamma = \frac{\mathcal C}{T_cB}
\end{equation}
where $\mathcal C$ denotes the total capacity per chip and $B$ denotes the bandwidth of the chip pulse.
Note that for Nyquist sinc pulses $T_cB=1$, while in general $T_cB$ can be either larger, e.g.,\ for root-raised cosine pulses, or smaller, i.e.,\ for anti-resource pooling, than 1.

The expression of the total capacity per chip for asynchronous
CDMA systems constrained to a given chip pulse waveform $\psi(t)$ of bandwidth $B$ and a given receive filter $g(t)$
 can be obtained by making use of the results in
Section \ref{sec:MMSE_det} and the fundamental relation between
mutual information and MMSE in Gaussian channels provided in
\cite{guo:05}.  Since such constrained total capacity depends on $\psi(t)$ and $g(t)$ only via the waveform $\phi(t),$ output of the filter $g(t)$ for the input $\psi(t),$ we shortly refer to it as the total capacity constrained to the chip waveform $\phi(t).$

\begin{corollary}\label{theor:constrained_capacity}
Let us adopt the same definitions as in Theorem
\ref{theo:small_bandwidth_MMSE} and let the assumptions of
Corollary \ref{cor_MMSE_raised_cosine} or Theorem
\ref{theo:small_bandwidth_MMSE} be satisfied.
 Additionally, let the receive filter and sampling process be such that sufficient discrete-time statistics are provided.
Then, as $K,N
\rightarrow \infty$ with $\frac{K}{N} \rightarrow \beta$
  the total capacity per chip
constrained to the chip pulse waveform $\phi(t)$ converges to the
deterministic value
\begin{equation}\label{relation_capacity_Stieltjes_transform}
\mathcal{C}^{({\rm asyn})}\left(\beta, \frac{E_{\phi}}{N_0}, \phi
\right)= \frac{\beta }{\ln2}  \int_{0}\limits^{
\frac{E_{\phi}}{N_0}}  \int \frac{\lambda
\eta_{\gamma} \mathrm{d}F_{|\matA|^2}(\lambda)}{1+
\lambda \gamma \eta_\gamma}\,\mathrm{d}\gamma
\end{equation}
where $\eta_\gamma$ is the multiuser efficiency at signal-to-noise ratio $\gamma$ given in
(\ref{SINR_cor_MMSE_raised_cos}) and
(\ref{SINR_theo_small_bandwidth}), respectively.
\end{corollary}
The proof of this corollary is discussed in Appendix
\ref{section:proof_theor_constrained_capacity}.

Let us consider again the case of sinc chip waveforms
as defined in  (\ref{sinc_waveform}) and uniform distribution of the
time delays.
Let $\alpha$ denote the bandwidth of the sinc pulse relative to the Nyquist bandwidth.
As noticed in Section \ref{sec:MMSE_det}, the multiuser efficiency $\eta_{\mathrm{sinc}}$ of an asynchronous  system with such sinc waveforms  given by (\ref{fix_point_aux}) and load $\beta$ equals
the multiuser efficiency $\eta_{\mathrm{syn}}$ of a synchronous system with Nyquist sinc pulses given by (\ref{fixed_point_MMSE_sinc}) and load  $\beta^\prime = \frac{\beta}{\alpha}.$
Since the load enters capacity per chip \eq{relation_capacity_Stieltjes_transform} only via the multiuser efficiency except for the linear pre-factor to the integral, we immediately find the following equation relating the two capacities per chip
\begin{equation}\label{eq:sinc_capacity}
 \mathcal{C}^{\text{(sinc)}} (\beta, \text{SNR},
\alpha)= \alpha\,
\mathcal{C}^{(\text{syn})} \left(\frac{\beta}{\alpha}, \text{SNR}
\right).
\end{equation}
 It is apparent from
(\ref{eq:sinc_capacity}) that synchronous and asynchronous systems
have the same capacity for $\alpha=1.$

In order to compare different systems (with possibly different
spreading gains and data rates), spectral efficiency has to
be given as a function of $\frac{E_b}{N_0},$ the level of energy
per bit per noise level equal to \cite{verdu:99a} \cite{shitz:99}
\begin{equation}
\frac{E_b}{N_0}= \frac{\beta \text{SNR}}{\mathcal{C}^{(*)}(\beta,
\text{SNR}, \cdot )}.
\end{equation}

\begin{figure*}
\noindent \begin{minipage}[t]{.49\linewidth} \begin{minipage}[t]{\linewidth} %
\centering \epsfig{file=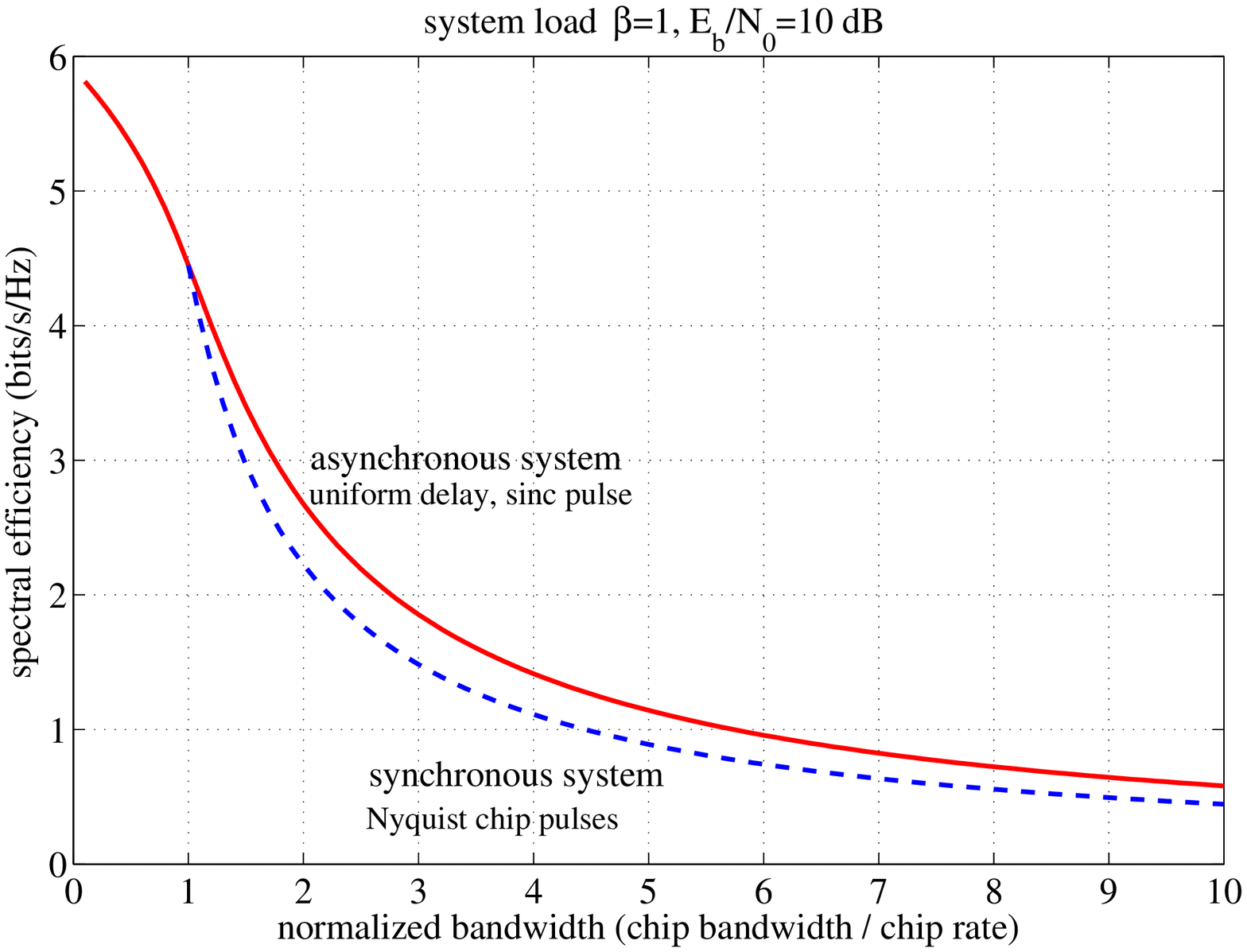,width=\linewidth}

{\caption{ Spectral efficiency of random CDMA with unit load versus
the normalized bandwidth $\alpha $ and $\frac{E_b}{N_0}=10 $ dB.
}\label{fig:capacity_vs_frequency}}
\end{minipage} \par\vspace*{0mm}\end{minipage}
\hfill \vspace*{0cm}
\begin{minipage}[t]{.49\linewidth}
\centering \epsfig{file=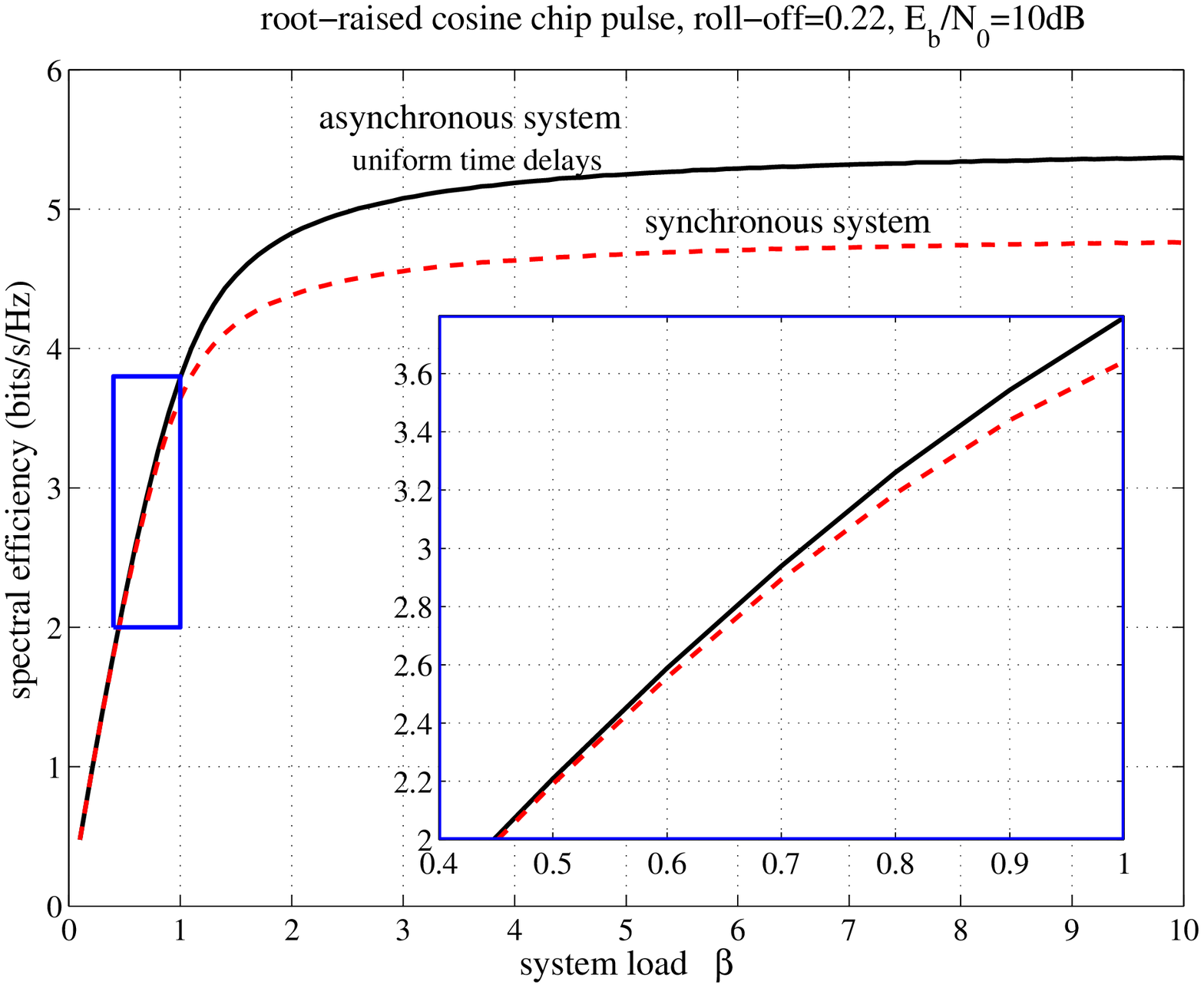,width=\linewidth}

{\caption{ Spectral efficiency of random CDMA versus
the load $\beta $ for the root-raised cosine chip pulse used in the UMTS standard and $\frac{E_b}{N_0}=10 $ dB.
}\label{fig:capacity_vs_load}}
\end{minipage}\end{figure*}

In Figure \ref{fig:capacity_vs_frequency}, we compare the spectral efficiency
of asynchronous CDMA with
the spectral efficiency of
synchronous CDMA.
TheF
spectral efficiencies are plotted against the bandwidth normalized to the Nyquist bandwidth with
$\frac{E_b}{N_0}=10 \, \mathrm{dB}$ and unit load $\beta=1.$
Recall from earlier discussions that for synchronous systems all Nyquist chip waveforms perform identically.
So there is no need to specify a particular Nyquist pulse except for the Nyquist pulse having the same bandwidth than the sinc pulse in the asynchronous case.
We see further that the smaller the normalized bandwidth, the higher the spectral efficiency is.
This is, as anti-resource pooling improves spectral efficiency by emulating a higher load.

In Figure \ref{fig:capacity_vs_load} the spectral efficiency is
plotted against the load $\beta$ with $\frac{E_b}{N_0}=10 \,
\mathrm{dB}$ for the chip waveform used in the UMTS standard.
When the load $\beta$
increases the gap in spectral efficiency between synchronous and asynchronous
systems increases.

\section{Extension to General Asynchronous CDMA Systems}\label{conj}

In this section we extend the previous results to any distribution
of the time delays for CDMA systems. Without loss of generality we
can assume that the time delays $\tau_k \in [0, T_s)$
\cite{verdu:98}. In this case, intersymbol interference is not
negligible and an infinite observation window is necessary to
obtain sufficient statistics. Equation (\ref{matrix_model}) for
the chip asynchronous but symbol quasi-synchronous  system model is extended to a general asynchronous system by

\begin{align}\label{general_discrete_output}
   y[p]=\sum_{k=1}^{K}a_k  \sum_{m=-\infty}^{+\infty} b_k[m]
   \overline{c}_{k}^{(m)}\left( \frac{p}{r}T_c - \tau_k \right)+w[p]
\end{align}
with $p \in \mathbb{Z}$ and
\begin{equation}
\overline{c}_{k}^{(m)}=\sum_{u=0}^{N-1} s_{k,m}[u] {\phi}
\!\left(\!t-\left(\!u\!+\!mN \right)\!T_c \right).
\end{equation}

By assuming the same approximation as in (\ref{matrix_model}), the
virtual spreading sequence of user $k$ in  the symbol interval $m$
has nonzero elements only in the time interval $m$ and $m+1.$ Let
$\overline{\tau}_k$ denote the delay of the signal $k$ in terms
of the chip intervals and $\widetilde{\tau}_k$ the delay within a
chip, i.e., $\overline{\tau}_k=\left \lfloor \frac{\tau_k}{T_c}
\right \rfloor$ and $\widetilde{\tau}_k=\tau_k \mathrm{mod} T_c,$
respectively. The virtual spreading sequence of user $k$ is
obtained by computing $\boldsymbol{\Phi}_k$ as in
(\ref{matrice_Phi}) for $\tau=\widetilde{\tau}_k$ to account for
the delay within a chip and then by shifting the virtual spreading
vector down by $\overline{\tau}_k$ $r$-dimensional blocks to
account for the delay multiple of the chip interval. More
precisely, the virtual spreading in the $m$-th symbol interval is
given by the $2rN$-dimensional vector
\begin{equation}\label{virtual_spreading_esteso}
\left[ \begin{array}{c}
 \Mat{0}_{\overline{\tau}_k} \\
 \mathfrak{F}(\Vec{c}(\widetilde{\tau}_k) \\
 \Mat{0}_{N-\overline{\tau}_k}
\end{array} \right]\vecs_k^{(m)} =\widetilde{\boldsymbol{\Phi}}_k \vecs_k^{(m)}
\end{equation}
with
\begin{equation*}
\Vec{c}(\widetilde{\tau}_k)=\left[ \Phi(\Omega,
\widetilde{\tau}_k), \Phi\left(\Omega,
\widetilde{\tau}_k-\frac{T_c}{r} \right), \ldots \Phi\left(\Omega,
\widetilde{\tau}_k-\frac{(r-1)T_c}{r}\right) \right],
\end{equation*}
$\Vec{0}_{\overline{\tau}_k}$ and $\Vec{0}_{N-\overline{\tau}_k}$
column vectors with zero entries and dimension $r\overline{\tau}_k$
and $r(N-\overline{\tau}_k), $ respectively. The $2rN \times K$
virtual spreading matrix for the symbols transmitted at time
interval $m$ is then
\begin{equation*}
  \widetilde{\matS}^{(m)}=\left[\widetilde{\boldsymbol{\Phi}}_1 \vecs_1^{(m)}, \widetilde{\boldsymbol{\Phi}}_2 \vecs^{(m)}_2, \ldots \widetilde{\boldsymbol{\Phi}}_K \vecs_K^{(m)}
  \right].
\end{equation*}
For further study, we introduce the upper and lower part of the
matrix $\widetilde{\matS}^{(m)},$ $\widetilde{\matS}^{(m)}_u$ and $\widetilde{\matS}^{(m)}_d$ of size
$rN \times rK$ such that
\begin{equation*}
  \widetilde{\matS}^{(m)}=\left[\begin{array}{c}\widetilde{\matS}^{(m)}_u \\ \widetilde{\matS}^{(m)}_d
  \end{array}
  \right].
\end{equation*}
and the matrices $\widetilde{\matH}_u^{(m)}=\widetilde{\matS}_u^{(m)}\matA$ and
$\widetilde{\matH}_d^{(m)}=\widetilde{\matS}_d^{(m)}\matA.$ Then, the baseband discrete-time
asynchronous system in matrix notation is given by
\begin{equation}\label{matrix_model_general}
\boldsymbol{\mathcal{Y}}=
\boldsymbol{\mathcal{H}}\boldsymbol{\mathcal{B}}+\boldsymbol{\mathcal{W}}
\end{equation}
where
$\stacky=[\ldots,\vecy^{(m-1)^T},\vecy^{(m)^T},\vecy^{(m+1)^T}
\ldots]^T$ and
$\stackb=[\ldots,\vecb^{(m-1)^T},\vecb^{(m)^T},\vecb^{(m+1)^T}
\ldots]^T$ are the infinite-length vectors of received and
transmitted symbols respectively; $\stackW$ is an infinite-length
white Gaussian noise vector; and $\stackH$ is a bi-diagonal block
matrix with infinite block rows and block columns
\begin{equation}\label{unlimited_stack_H} \stackH = \left[ \begin{array}{ccccccc} \ddots & \ddots
&\ddots &\ddots &\ddots &\ddots &\ddots
\\
  \ldots & \mathbf{0} & \widetilde{\matH}_d^{(m-1)} & \widetilde{\matH}_u^{(m)} & \mathbf{0} &
 \ldots & \ldots  \\  \ldots & \ldots  & \mathbf{0} & \widetilde{\matH}_d^{(m)} & \widetilde{\matH}_u^{(m+1)} & \mathbf{0} &
 \ldots   \\ \ddots & \ddots &\ddots
&\ddots &\ddots &\ddots &\ddots
\\
\end{array} \right].
\end{equation}
Finally, we define the correlation matrices
$\stackT=\stackH\stackH^H$, $\stackR=\stackH^H \stackH$.

The following theorem shows that a linear MMSE
detector for a CDMA system with transfer matrix
${\stackH}$ and time delays $\tau_1, \tau_2, \ldots
\tau_K$ has the same limiting performance as a linear MMSE
detector for chip asynchronous but symbol quasi-synchronous CDMA
systems  introduced in Section \ref{sec:system_model}  with time delays $\widetilde{\tau}_1, \widetilde{\tau}_2,
\ldots \widetilde{\tau}_K.$ The same equivalence holds for capacity and spectral efficiency.

\begin{theor}\label{theo:equivalence_widetildestackH}
Given $\{\tau_1, \tau_2, \ldots \tau_K \}$ a set of delays in $[0,
T_s)$ let us consider the set of delays in $[0, T_c)$ defined as
$\left\{\widetilde{\tau}_k: \; \widetilde{\tau}_k= \tau_k  \mod T_c, \; k=1, \ldots K
\right\}.$ Given a positive integer $r,$  let $\Mat{\Phi}_k, \, k=1,\ldots K,$ be the $r$-blockwise circulant matrix of order $\N$ defined in (\ref{matrice_Phi}) with $\tau=\widetilde{\tau}_k.$   Let $\matA,$ $\Phi(\omega),$
 ${\matS},$ and ${\matH}$ be defined as in Theorem
\ref{theo:sinr_MMSE_chip_asynch}. Furthermore,
$\widetilde{\Mat{\Phi}}_k,$ $k=1 \ldots K $ are $2 r N \times N$
matrices such that $\widetilde{\Mat{\Phi}}_k= [\Mat{0}_{\overline{\tau}_k}^T,
\boldsymbol{\Phi}_k^{T}, \Mat{0}_{N-\overline{\tau}_k}^T
]^T$ with $\overline{\tau}_k=\left\lfloor \frac{\tau_k}{T_c} \right\rfloor,$
$\Mat{0}_{\overline{\tau}_k}$ and $\Mat{0}_{N-\overline{\tau}_k}$ zero matrices of dimensions $
r \overline{\tau_k} \times N$ and
$r\left( N-  \overline{\tau_k}\right) \times N,$ respectively. Let $\widetilde{{\matS}}^{(m)}=\left(
\widetilde{\Mat{\Phi}}_1 \vecs_{1}^{(m)}, \widetilde{\Mat{\Phi}}_2
\vecs_{2}^{(m)} \ldots \widetilde{\Mat{\Phi}}_K \vecs_{K}^{(m)} \right),$
$\widetilde{\matH}^{(m)}=[\widetilde{\matH}_u^{(m)T}, \widetilde{\matH}_d^{(m)T}]^T = \widetilde{\matS} \matA$ and ${\stackH}$ the infinite block row and block column
matrix of the same form as in (\ref{unlimited_stack_H}).
Let the same assumptions as in Theorem \ref{theo:sinr_MMSE_chip_asynch} hold.

Then, asymptotically, as $K, N \rightarrow \infty$ with
$\frac{K}{N} \rightarrow \beta$ the CDMA systems transfer matrices ${\stackH}$ and ${\matH}$ are equivalent in terms of multiuser efficiency for linear MMSE detectors and in terms of spectral efficiency.
\end{theor}
This theorem is shown in Appendix
\ref{section:proof_theo_equivalence_widetildestackH}.

Interestingly, the system performance depends on the time delays
${\tau}_k$ only through the offsets $\widetilde{\tau}_k - \left
\lfloor \frac{\tau_k}{T_c} \right\rfloor T_c.$ Therefore, any
shift of the signal multiple of $T_c$ does not affect the
performance of the system.

The analysis presented in this contribution has been restricted to frequency flat fading for the sake of clarity. The extension to multipath fading channels is straightforward when the impulse response of the channel is much shorter than the symbol interval. In fact, the chip pulse waveform at the output of the matched filter $\phi(t)$ can include the effects of the frequency selective channel impulse response $a(t)$ along with the effects of the transmitted chip pulse waveform $\psi(t),$ and the filter at the front-end $g(t)$, i.e., $\phi(t)= \psi(t)*a(t)*g(t).$   Then, the analysis of a system with frequency selective fading reduces to the proposed analysis.


\section{Conclusions}\label{chap:async_sec:conclusions}

This work provides a general framework for the analysis of
asynchronous CDMA systems with random spreading using sufficient
or suboptimum statistics and any chip pulse waveform. Furthermore,
it includes several optimum or suboptimum receiver structures of
practical and theoretical interest. Therefore, it provides insight
into both the fundamental limits of asynchronous CDMA systems and
the performance loss of implementations where suboptimum
receiver structures, suboptimum statistics, and/or non-ideal chip
pulses are utilized.

For the receiver structures investigated in Part I,  the
performance of a CDMA system is independent of the time delay
distribution if the bandwidth of the chip pulse waveform is not
greater than half of the chip rate, i.e., $B \leq \frac{1}{2 T_c}.$
This also implies that synchronous and asynchronous CDMA systems
have the same performance and generalizes the equivalence result
in \cite{mantravadi:02} for Nyquist sinc ($B = \frac{1}{2 T_c}$)
pulses and linear MMSE detectors to any chip pulse waveform.
The behavior of CDMA system changes substantially as  the
bandwidth gets larger. In this case, the system performance is
significantly affected by the distribution of the time delays and
the performance of linear detectors may depend on the specific
time delay of the signal of interest. If the receiver is fed by
sufficient statistics and the time delay distribution is uniform
the performance of optimum or suboptimum receivers is independent
of the time delays.
In the following,
we summarize the most interesting aspects pointed out by the large
system analysis, for each class of receivers.

\subsection{Optimum Receiver}

The spectral efficiency constrained to a given chip pulse waveform
characterizes the performance of a CDMA channel with optimum
receiver. The spectral efficiency is
expressed in terms of the multiuser efficiency spectral density $\eta(\omega).$
When the chip-modulation is based
on sinc pulses whose bandwidth is $\alpha$ times the Nyquist bandwidth, the spectral efficiency of asynchronous CDMA systems is
identical to the spectral efficiency of synchronous systems with load
$\beta^{\prime}=\frac{\beta}{\alpha}$ and Nyquist sinc pulses.
Spectral efficiency is a strictly decreasing function of the relative pulse bandwidth $\alpha$ and for $\alpha\to0$, the spectral efficiency of a single user AWGN channel is reached.

 For $\alpha > 1$ an asynchronous CDMA system
with modulation based on a sinc function can compensate to some
extent for the loss in spectral efficiency of synchronous CDMA
systems with equal bandwidth. For $\beta \rightarrow \infty$ it
attains the maximum spectral efficiency for any finite bandwidth
$B=\frac{\alpha}{2 T_c}.$

\subsection{Linear MMSE Detector}

The output SINR of a linear MMSE detector can be obtained from the
solution to a system of fixed point equations in the general case.
In the two cases (i) chip pulses with
bandwidth $B \leq \frac{1}{2T_c}$  and (ii) chip pulses with
bandwidth $B>
\frac{1}{2 T_c}$, sufficient statistics and uniform time delay distribution
the fixed point system of equation reduces to a single equation.
In those cases, the performance of a linear MMSE detector in
asynchronous CDMA systems is characterized by a unique value  of
multiuser efficiency. Furthermore, the measure of multiuser
efficiency can be refined by the concept of spectrum of the
multiuser efficiency that is also unique for all the users.
Additionally, for these CDMA systems the limiting interference
effects can be decoupled into
user domain and frequency domain, as the system grows large, such that we can define an
effective interference spectral density similarly to the effective
interference in \cite{tse:99b} for synchronous systems.

In the special case that the modulation is based on sinc functions
with bandwidth $B=\frac{\alpha}{2T_c},$ a linear MMSE detector in
asynchronous CDMA channels performs identically to a synchronous CDMA
system with square root Nyquist chip pulses \cite{tse:99b} and
load $\beta^{\prime}=\frac{\beta}{\alpha}.$ This effect is similar
to the resource pooling effect for synchronous CDMA systems with
spatial diversity in \cite{hanly:99a} and shows the possibility to
trade degrees of freedom in the frequency domain against degrees
of freedom in the time domain.

Though this work focused on performance measures for CDMA, similar results hold for asynchronous MIMO systems due to the mathematical analogy between CDMA and MIMO systems when described as a discrete-time vector channel.
This means, that MIMO systems with excess bandwidth and desynchronized modulators for different antenna elements benefit in a similar manner than CDMA systems with desynchronized users.

\section*{Acknowledgment}

The authors thank Alex Grant for useful discussions.

\appendices

\section{Useful Mathematical Tools}\label{section:mathematical_tools}
Let $\Phi(\omega)$ be the unitary Fourier transform of a pulse waveform
$\phi(t)$ with bandwidth $B \leq \frac{r}{2T_c} .$ Then, in the
normalized frequency interval $\Omega \in \left[-\pi,
\pi \right]$ the unitary Fourier transform
(\ref{discrete_Fourier_phi}) of the sequence obtained by sampling
$\phi(t)$ at time instant $\tau$ and rate $\frac{r}{T_c}$ is given
by
\begin{equation}\label{phi_form_bandlimited}
{\phi}(\Omega,\tau)= \frac{1}{T_c} \mathrm{e}^{   \frac{
j\tau}{T_c}\Omega} \sum_{s \in \mathcal{Z}(\Omega)}
\mathrm{e}^{  j 2 \pi \frac{ \tau}{T_c}s} \Phi^{*} \left(\frac{\Omega+ 2
\pi s}{T_c} \right) \quad \text{for } \quad |\Omega|\leq
\pi
\end{equation}
where $\mathcal{Z}(\Omega)$ is the set of all integers in the interval\\ $\left[\min\left(-\mathrm{sign}(\Omega)\left\lfloor \frac{r-1}{2}
\right\rfloor
, \mathrm{sign}(\Omega)\left\lfloor\frac{r}{2}\right\rfloor\right), \;\max\left(-\mathrm{sign}(\Omega)\left\lfloor \frac{r-1}{2}
\right\rfloor
, \mathrm{sign}(\Omega)\left\lfloor\frac{r}{2}\right\rfloor\right)\right].$

The matrix
\begin{equation}\label{matrixQ(xt)_def}\Mat{Q}(\Omega,\tau)=\Mat{\Delta}_{\phi,r}(\Omega,\tau)
\Mat{\Delta}_{\phi,r}(\Omega,\tau)^H,\end{equation}  with
$\Mat{\Delta}_{\phi,r}(\Omega,\tau)$ defined in
(\ref{delta_phi_r_def}), can be decomposed in the sum of two
matrices
\begin{equation}\label{Q(xt)_decomposition}
\Mat{Q}(\Omega,\tau)=\Mat{Q}(\Omega)+\overline{\Mat{Q}}(\Omega,\tau)
\end{equation}
where the $(k, \ell)$-elements of the matrices $\Mat{Q}(\Omega)$ and
$\overline{\Mat{Q}}(\Omega,\tau)$ are given by
\begin{equation}\label{elem_Q(x)}
(\Mat{Q}(\Omega))_{k,\ell}=\frac{1}{T_c^2}
\sum_{s \in \mathcal{Z}(\Omega)}
\left|\Phi\left(\frac{\Omega+ 2 \pi s}{T_c} \right)\right|^2
\mathrm{e}^{- j\frac{ k-\ell}{r}(\Omega+2 \pi s)}  \quad \text{for }
\quad |\Omega|\leq \pi,
\end{equation}
and
\begin{multline}\label{elem_Q(x_tau)}
(\overline{\Mat{Q}}(\Omega,\tau))_{k,\ell}=\frac{1}{T_c^2}
\sum_{\begin{subarray}{c}
  s,u \in \mathcal{Z}(\Omega) \\
  s \neq u
\end{subarray}}\Phi\left(\frac{\Omega+ 2 \pi u}{T_c} \right)\Phi^{*}\left(\frac{\Omega+ 2
\pi s}{T_c} \right) \mathrm{e}^{-j 2 \pi \frac{
\tau}{T_c}(s-u)} \mathrm{e}^{-j  \left(\frac{
k-1}{r}(\Omega-2 \pi s)-\frac{ \ell-1}{r}(\Omega-2 \pi u) \right)} \\ \text{for } \quad
|\Omega|\leq \pi,
\end{multline}
respectively.

Useful properties of the matrices $\Mat{Q}(\Omega)$ and
$\overline{\Mat{Q}}(\Omega,\tau)$ are stated in the following lemmas.

\begin{lemma}\label{Lemma:traceBQ}
Let $\Mat{B}$ be an $r \times r$ matrix of the form
\begin{equation}\label{matrixBform}
\Mat{B}= \Mat{B}(\Omega)=\left[ \begin{array}{ccccc}
  b_0 & b_1 \mathrm{e}^{j \frac{\Omega}{r} } & \ldots & \ldots & b_{r-1} \mathrm{e}^{j \frac{ (r-1)}{r} \Omega} \\
  b_{r-1} \mathrm{e}^{-j \frac{\Omega}{r}} & b_0 & b_1 \mathrm{e}^{j \frac{\Omega}{r} } & \ldots & b_{r-2} \mathrm{e}^{j \frac{ (r-2)}{r} \Omega} \\
  \ldots & \ddots & \ddots & \ddots & \ddots \\
  b_{1} \mathrm{e}^{-j \frac{ (r-1)}{r} \Omega} & \ddots & \ddots & b_{r-1} \mathrm{e}^{-j \frac{\Omega}{r} } & b_0
\end{array} \right],
\end{equation}
i.e., given $b_0=b_0(\Omega),b_1=b_1(\Omega), \ldots b_{r-1}=b_{r-1}(\Omega),$
eventually functions of $\Omega,$ $(\Mat{B})_{\ell,k},$ the element
$(\ell,k)$ of the matrix $\Mat{B}$ satisfies
$(\Mat{B})_{\ell,k}=\mathrm{e}^{\frac{j(k-\ell)}{r}\Omega}
b_{(r+k-\ell)\mathrm{mod}r}.$ Let $\overline{\Mat{Q}}(\Omega,\tau)$ be
the $r \times r$ matrix with element $(k, \ell)$ defined in
(\ref{elem_Q(x_tau)}). Then,
\begin{equation*}
\mathrm{tr}(\Mat{B \overline{Q}}(\Omega, \tau))=0.
\end{equation*}
\end{lemma}
\begin{proof}
Let $\overline{q}_{us}(\Omega)=\frac{1}{T_c^2}\Phi\left(\frac{\Omega+ 2
\pi u}{T_c} \right)\Phi^{*}\left(\frac{\Omega+ 2
\pi s}{T_c}\right).$ Then,
\begin{align}
\mathrm{tr}(\Mat{B \overline{Q}}(\Omega+j 2
\pi u, \tau))&= \sum_{k=1}^{r}
\sum_{\ell=1}^{r}  (\Mat{\overline{Q}}(\Omega,
\tau))_{k,\ell}(\Mat{B})_{\ell,k} \nonumber\\
&= \sum_{\begin{subarray}{c}
  s,u \in \mathcal{Z}(\Omega) \\
  s \neq u
\end{subarray}}\overline{q}_{us} \mathrm{e}^{j 2 \pi \frac{ \tau}{T_c}(s-u)}
\sum_{k=1}^{r} \sum_{\ell=1}^{r} (\Mat{B})_{\ell,k} \mathrm{e}^{-j
 \frac{ k-\ell}{r} \Omega } \mathrm{e}^{-j \frac{ 2
\pi}{r}(-s(k-1)+u(\ell-1))} \nonumber \\
&= \sum_{\begin{subarray}{c}
  s,u \in \mathcal{Z}(\Omega) \\
  s \neq u
\end{subarray}}\overline{q}_{us} \mathrm{e}^{j 2 \pi \left(\frac{
\tau}{T_c}-\frac{1}{r} \right)(s-u) }  \sum_{k=1}^{r}
\sum_{\ell=1}^{r} b_{(r+k-\ell)\mathrm{mod}r} \mathrm{e}^{-j
\frac{ 2 \pi}{r}(u-s)k} \mathrm{e}^{j \frac{ 2
\pi}{r}u(k-\ell)} \nonumber \\
&=\sum_{\begin{subarray}{c}
  s,u \in \mathcal{Z}(\Omega) \\
  s \neq u
\end{subarray}}\overline{q}_{us} \mathrm{e}^{j 2 \pi \left(\frac{
\tau}{T_c}-\frac{1}{r} \right)(s-u)} (\eta_{1}+ \eta_{2})
\nonumber
\end{align}
with
\begin{equation}\label{eta1}
\eta_1=\sum^{r}_{\begin{subarray}{c}
  {k,\ell=1} \\
  k \geq \ell
\end{subarray}} b_{(r+k-\ell)\mathrm{mod}r} \mathrm{e}^{-j \frac{ 2
\pi}{r}(u-s)k} \mathrm{e}^{j \frac{ 2 \pi}{r}u(k-\ell)}
\end{equation}
and
\begin{equation}\label{eta2}
\eta_2=\sum^{r}_{\begin{subarray}{c}
  {k,\ell=1} \\
  k < \ell
\end{subarray}} b_{(r+k-\ell)\mathrm{mod}r} \mathrm{e}^{-j \frac{ 2
\pi}{r}(u-s)k} \mathrm{e}^{j \frac{ 2 \pi}{r}u(k-\ell)}.
\end{equation}
Substituting $v=k-\ell$ in (\ref{eta1}) and $v= r+k-\ell$ in
(\ref{eta2}) we obtain
\begin{equation*}
\eta_1=\sum_{k=1}^{r}  \sum_{v=0}^{k-1} b_{v} \mathrm{e}^{-j
\frac{ 2 \pi}{r}(u-s)k} \mathrm{e}^{j \frac{ 2 \pi u v}{r}}
\end{equation*}
and
\begin{equation*}
\eta_2=\sum_{k=1}^{r}  \sum_{v=k}^{r-1} b_{v} \mathrm{e}^{-j
\frac{ 2 \pi}{r}(u-s)k} \mathrm{e}^{j \frac{ 2 \pi u v}{r}},
\end{equation*}
respectively. For $s,t \in  \mathcal{Z}(\Omega) $ and $s
\neq t$, $|s-t| \in [1, \ldots, r-1].$ Therefore, $\sum_{k=1}^r
\mathrm{e}^{-j \frac{ 2 \pi}{r}(u-s)k}=0$ and $\eta_1+\eta_2=0$
for all $\Omega.$ Then, also $\mathrm{tr}(\Mat{B \overline{Q}}(\Omega,
\tau))=0$ and this concludes the proof of Lemma
\ref{Lemma:traceBQ}.
\end{proof}
It follows immediately from Lemma \ref{Lemma:traceBQ} that
$\mathrm{tr}\Mat{ \overline{Q}}(\Omega, \tau)=0$ since the identity
matrix $\Mat{I}$ is of the form (\ref{matrixBform}) with $b_0=1$
and $b_i=0$ for $i=1, \ldots r-1.$

\begin{lemma}\label{lemma:invarianceQT}
Let $\Mat{B}=\Mat{B}(\Omega)$ be a matrix defined as in Lemma
\ref{Lemma:traceBQ} and let $\Mat{Q}(\Omega)$ be the $r \times r$
matrix with element $(k,\ell)$ defined in (\ref{elem_Q(x)}). Then,
the matrix $\Mat{C}(\Omega)=\Mat{Q}(\Omega) \Mat{B}(\Omega)$ is of the form
(\ref{matrixBform}).
\end{lemma}
\begin{proof}
The element $(k, \ell)$ of the matrix $\Mat{C}=\Mat{C}(\Omega),$
$(\Mat{C})_{\ell,k}$ is given by
\begin{align}
(\Mat{C})_{\ell,k} & =\sum_{t=1}^{r} (\Mat{B})_{\ell, t}
(\Mat{Q}(\Omega))_{t,k} \nonumber\\
& = \mathrm{e}^{j \frac{k-\ell}{r} \Omega} \kappa(\ell,k)
\end{align}
with
\begin{equation*}
\kappa(\ell,k) =
\frac{1}{T_c^2}\sum_{s \in \mathcal{Z}(\Omega)}
\left|\Phi\left( \frac{\Omega +2 \pi s}{T_c} \right)\right|^2
\eta(\ell, k,s)
\end{equation*}
and
\begin{equation}\label{aux_var_lemma2}
\eta(\ell,k,s)=\sum_{t=1}^{r} b_{(r+t-\ell)\mathrm{mod}r}
\mathrm{e}^{-j 2 \pi \left( \frac{t-k}{r}\right)s}.
\end{equation}
In order to prove Lemma \ref{lemma:invarianceQT} it is sufficient
to prove that
\begin{equation}\label{property_kappa}\kappa(\ell,
k)=\kappa((\ell+1)\mathrm{mod}r, (k+1)\mathrm{mod}r).
\end{equation} In fact, in
this case $\Mat{C}_{\ell,k}=\mathrm{e}^{ j \frac{k-\ell}{r}
\Omega} \kappa_{(r+k-\ell)\mathrm{mod}r}$ with
$\kappa_{(r+k-\ell)\mathrm{mod}r}=\kappa(\ell,k).$ The property
(\ref{property_kappa}) is implied by a similar property on
$\eta(\ell, k,s)$
\begin{equation}\label{property_eta}\eta(\ell,
k,s)=\eta((\ell+1)\mathrm{mod}r,
(k+1)\mathrm{mod}r,s).\end{equation} It is straightforward to
verify that (\ref{property_eta}) is satisfied since both factors
$b_{(r+t-\ell)\mathrm{mod}r}$ and $\mathrm{e}^{-j 2\pi \left(
\frac{t-k}{r}\right)}$ are periodical in their arguments $\ell$
and $k$, respectively, with period $r$ and $k$ and $\ell$ are
simultaneously increased by a unit. This concludes the proof of
Lemma \ref{lemma:invarianceQT}.
\end{proof}
The following lemma provides the eigenvalue decomposition of the
matrix $\Mat{Q}(\Omega).$
\begin{lemma}\label{lemma:Q(x)_decomposition}
Let $\Mat{Q}(\Omega)$ be an $r \times r$ matrix with element $(k,\ell)$
defined in (\ref{elem_Q(x)}). Then, the matrix $\Mat{Q}(\Omega)$ can be
decomposed as follows
\begin{equation}\label{Q(x)_decomposition}
\Mat{Q}(\Omega)=\Mat{U}(\Omega) \Mat{D}(\Omega)\Mat{U}^H(\Omega)
\end{equation}
where
\begin{equation}\label{U_definition}
\Mat{U}(\Omega)=\left( \Vec{e}\left(\Omega-\mathrm{sign}(\Omega)2 \pi \left \lfloor
\frac{r-1}{2} \right \rfloor \right), \ldots \Vec{e}\left(\Omega
\right) \ldots \Vec{e}\left( \Omega+\mathrm{sign}(\Omega) 2 \pi \left \lfloor
\frac{r}{2} \right \rfloor \right) \right),
\end{equation}
$\Vec{e}\left( \Omega \right)$ is an r-dimensional column vector
defined by
\begin{equation*}
\Vec{e}\left( \Omega \right)= \frac{1}{\sqrt{r}}\left( 1,
\mathrm{e}^{-j\frac{\Omega}{r}}, \ldots \mathrm{e}^{-j\frac{r-1}{r} \Omega} \right)^T,
\end{equation*}
and $\Mat{D}(\Omega)$ is the diagonal matrix whose $s^{\mathrm{th}}$
diagonal element is given by
\begin{equation}\label{elements_D(x)}
(\Mat{D}(\Omega))_{ss}= \frac{r}{T_c^2} \left| \Phi\left(  \frac{1
}{T_c} \left( \Omega- \mathrm{sign}(\Omega) 2 \pi \left(\left\lfloor
\frac{r-1}{2}\right\rfloor-s+1 \right)\right) \right) \right|^2.
\end{equation}
\end{lemma}
\begin{proof}
Decomposition (\ref{Q(x)_decomposition}) can be immediately
derived by noting that
\begin{equation*}
\Mat{Q}(\Omega)=\sum_{s \in \mathcal{Z}(\Omega)} \frac{r}{T_c^2} \left| \Phi\left( \frac{1}{T_c}\Omega+2 \pi
s \right) \right|^2 \Vec{e}(\Omega+2 \pi s)
\Vec{e}^{H}(\Omega + 2 \pi s).
\end{equation*}
This expression can be rewritten as (\ref{Q(x)_decomposition}) and
Lemma \ref{lemma:Q(x)_decomposition} is proven.
\end{proof}

The following lemma shows that the matrix $\Mat{Q}(\Omega)$ and any
other matrix with the same basis of eigenvectors is of the form
(\ref{matrixBform}).
\begin{lemma}\label{lemma:matrix_with_eigenbasis_U}
Let $\Mat{C}(\Omega)= \Mat{U}(\Omega) \Mat{M}(\Omega) \Mat{U}^{H}(\Omega)$ with
$\Mat{U}(\Omega)$ unitary matrix defined in (\ref{U_definition}) and
$\Mat{M}(\Omega)$ diagonal matrix with elements $m_{kk}(\Omega).$ Then,
$\Mat{C}(\Omega)$ is of the form (\ref{matrixBform}).
\end{lemma}
\begin{proof}
The $\ell^{\mathrm{th}}$ row of the matrix $\Mat{U}(\Omega)$ is given
by
\begin{equation*}
\Vec{u}_{\ell}(\Omega)= \frac{1}{\sqrt{r}} \left( \mathrm{e}^{-j
\frac{\ell-1}{r}\left(\Omega- \mathrm{sign}(\Omega) 2 \pi \lfloor
\frac{r-1}{2}\rfloor \right) }, \ldots \mathrm{e}^{-j
\frac{\ell-1}{r}\left(\Omega+ \mathrm{sign}(\Omega) 2 \pi \lfloor \frac{r}{2}
\rfloor \right) } \right)
\end{equation*}
and $c_{\ell k}(\Omega),$ the element $(\ell,k)$ of the matrix
$\Mat{C}$ satisfies
\begin{align}
c_{\ell k} (\Omega)&= \frac{1}{r} \sum_{i=1}^{r} m_{ii} \mathrm{e}^{-j
\frac{\ell-k}{r}\left(\Omega- \mathrm{sign}(\Omega) 2 \pi \lfloor
\frac{r-1}{2}\rfloor +2 \pi (i-1)\right) } \nonumber \\
& = \widetilde{b}_{\ell k} \mathrm{e}^{-j \frac{(\ell-k)}{r}
\Omega}
\end{align}
with $\widetilde{b}_{\ell k}= \sum_{i=1}^{r} \frac{m_{ii}}{r}
\mathrm{e}^{j 2 \pi \frac{\ell-k}{r} \left( \mathrm{sign}(\Omega)
\lfloor \frac{r-1}{2} \rfloor -i+1 \right) }.$ It is
straightforward to verify that $\widetilde{b}_{\ell k} =
\widetilde{b}_{(\ell+1) \mathrm{mod}r, (k+1) \mathrm{mod}r }.$
This concludes the proof of Lemma
\ref{lemma:matrix_with_eigenbasis_U}.
\end{proof}

The following lemmas state results from random matrix theory
developed along the lines of the REFORM method proposed by Girko
in \cite{girko:01a} and \cite{girko:01b}.

\begin{lemma}\cite{girko:01a,cottatellucci:07}\label{lemma_T_R}
Let $\matXi=(\xi_{ij})_{i=1,\ldots N q_1}^{j=1,\ldots K q_2}$ be
an $ N q_1 \times K q_2$ matrix of complex random elements
$\xi_{ij}$ structured in $NK$ blocks of size $ q_1 \times q_2$,
$\matXi_{ij}$, i.e., \begin{equation*} \matXi=(\matXi_{ij})_{i=1,
\ldots N}^{j=1, \ldots K}
\end{equation*}
and $K=\beta N$ with $\beta>0.$ Let
$\matPTilde=(\matP_{ij})_{ij=1,\ldots p_1}=[\matXi \matXi^H
+\alpha \I]^{-1}$ and $\matGTilde=(\matG_{ij})_{ij=1,\ldots
p_2}=[\matXi^H \matXi+\alpha \I]^{-1}$, where $\matP_{ij}$ and
$\matG_{ij}$ are complex blocks of size $q_1 \times q_1 $ and $q_2
\times q_2,$ respectively.

Additionally, assume
\begin{description}
\item[H-1]\label{hyp1} $\matXi_{ks}$, $k=1, \ldots, N$, $s=1, \ldots, K$,
the random blocks of the matrix $\matXi$ are independent.
\item[H-2] \label{hyp2} All the elements of the matrix $\matXi$ are zero mean, i.e., $\E\{\matXi\}=\Mat{0}.$
\item[H-3]\label{hyp3}  $\sup_{K,N} \max_{i=1, \ldots,
N} \sum_{j=1}^{K} \E \|  \matXi_{ij} \|^2 + \sup_{K,N} \max_{j=1,
\ldots, K} \sum_{i=1}^{N} \E \| \matXi_{ij} \|^2 < + \infty,$
\item[H-4] \label{hyp4} Lindeberg condition: $\forall \tau>0$
\begin{multline}
\lim_{K=\beta N \rightarrow \infty} \left( \max_{i=1, \ldots, N}
\sum_{j=1}^{K} \E \left(\|\matXi_{ij} \|^2 \chi\{\|\matXi_{ij} \|>
\tau \} \right) + \max_{j=1, \ldots, K} \sum_{i=1}^{N} \E
\left(\|\matXi_{ij} \|^2 \chi\{\|\matXi_{ij} \|> \tau \} \right)
\right)=0.
\end{multline}
\end{description}
Then, for $\alpha \in \mathbb{C} \backslash \mathbb{R}^{-}$
\begin{equation*}
\lim_{K=\beta N \rightarrow \infty} \E|\Mat{P}_{p \ell}(\alpha) -
\matT_{p \ell}(\alpha)|=0 \qquad p, \ell= 1, \ldots, p_1
\end{equation*}
and
\begin{equation*}
\lim_{K=\beta N \rightarrow \infty} \E |\Mat{G}_{p \ell}(\alpha) -
\alpha^{-1} \Mat{R}_{p \ell}(\alpha)|=0 \qquad p, \ell= 1, \ldots,
p_2
\end{equation*}
i.e., the blocks of the matrices $\matQTilde$ and $\matGTilde$
converge in the first mean to the corresponding blocks of the
matrices
\begin{equation*}
\matTTilde=\mathrm{diag}((\Mat{C}_{nn}^{(1)}(\alpha))^{-1})_{n=1,
\ldots, N}
\end{equation*}
and
\begin{equation*}
\matRTilde=\mathrm{diag}((\Mat{C}_{kk}^{(2)}(\alpha))^{-1})_{k=1,
\ldots, K}
\end{equation*}
respectively.  The matrix blocks $\Mat{C}_{nn}^{(1)}(\alpha)$ of
size $q_1 \times q_1$ and $\Mat{C}_{kk}^{(2)}(\alpha)$ of size
$q_2 \times q_2$ are equal to
\begin{align}
\Mat{C}_{nn}^{(1)}(\alpha) &= \alpha \Mat{I} + \sum_{j=1}^{K} \E
\left(\matXi_{nj} (\Mat{X} )_{jj} \matXi_{nj}^H \right)_{\Mat{X}=
\alpha \matGTilde} & n=1,\ldots, N  \label{C1_diagonal_blocks}\\
\Mat{C}_{kk}^{(2)}(\alpha) &= \Mat{I} + \sum_{j=1}^{p_1} \E \left(
\matXi_{j k}^H (\Mat{Y} )_{jj} \matXi_{j k}^H \right)_{\Mat{Y}=
\matPTilde} & k=1,\ldots, K, \label{C2_diagonal_blocks}
\end{align}
respectively.
\end{lemma}
\begin{lemma}\cite{girko:01a,cottatellucci:07}\label{lemma_convergence_canonical_eqns}
Let us assume that the definitions of Lemma \ref{lemma_T_R} hold
and the conditions of Lemma \ref{lemma_T_R} are satisfied.

Then, the  $q_1 \times q_1$ matrices $\Mat{C}^{(1)}_{nn}(\alpha)$,
$n=1,\ldots, N$ and the $q_2 \times q_2$ matrices
$\Mat{C}_{kk}^{(2)}(\alpha)$, $k=1, \ldots, K$, defined in
(\ref{C1_diagonal_blocks}) and (\ref{C2_diagonal_blocks}),
respectively, converge as $K = \beta N \rightarrow \infty$ to the
limit matrices
\begin{align}
\lim_{K=\beta N \rightarrow +\infty} \Mat{C}_{nn}^{(1)} & {=}
\boldsymbol{\Psi}_{nn}^{(1)}
\qquad n=1,\ldots,N \nonumber \\
\lim_{K=\beta N \rightarrow +\infty} \Mat{C}_{kk}^{(2)} & {=}
\boldsymbol{\Psi}_{kk}^{(2)} \qquad k =1,\ldots, K \nonumber
\end{align}
where $\boldsymbol{\Psi}_{nn}^{(1)}$, $k= 1, \ldots, N$ and
$\boldsymbol{\Psi}_{kk}^{(2)}$, $k= 1, \ldots, K$ satisfy the
canonical system of equations
\begin{align}\label{C_1girko}
    \boldsymbol{\Psi}_{nn}^{(1)}\!&=  \alpha  \Mat{I}\!+ \!\sum_{j=1}^{K} \! \mathrm{E}\!
    \left\{  \Mat{\Xi}_{nj} \left[
    {\boldsymbol{\Psi}}_{jj}^{(2)} \right]^{-1}\Mat{\Xi}_{nj} ^H \! \right\},&  n \!= \!1,\ldots,\! N, \\
   \boldsymbol{\Psi}_{kk}^{(2)} \!&=  \Mat{I} \! + \! \sum_{j=1}^{N}
\mathrm{E}\left\{\Mat{\Xi}_{jk}^H
\left[{\boldsymbol{\Psi}}^{(1)}_{jj} \right]^{-1} \Mat{\Xi}_{j
k}\right\} ,& k \! = \! 1,\ldots, K. \label{C_2girko}
\end{align}
\end{lemma}
The following Lemma states the existence and uniqueness of the
solution of the system of canonical equations in the class of
definite positive Hermitian matrices.
\begin{lemma}\cite{girko:01a}\label{lemma_existence_uniqueness}
Let us adopt the definitions of Lemma \ref{lemma_T_R} and let us
assume that  the conditions of Lemma \ref{lemma_T_R} are
satisfied. Let us consider the system of canonical equations
(\ref{C_1girko}) and (\ref{C_2girko}). Then, the solution of the
canonical system of equations (\ref{C_1girko}) and
(\ref{C_2girko}) exists and it is unique in the class of
nonnegative definite analytic matrices for
$\mathrm{Re}{(\alpha)}>0.$
\end{lemma}
The following lemma due to Girko provides convergence of the
eigenvalue distribution of the matrix $\matXi \matXi^H$ with
$\matXi$ defined in Lemma \ref{lemma_T_R} to a deterministic
distribution function and the corresponding  Stieltjes transform.
\begin{lemma}\cite{girko:01a} \label{lemma_convergence_Stieltjes_transform_girko}
Let us adopt the definitions in Lemma \ref{lemma_T_R} and let the
assumptions of Lemma \ref{lemma_T_R} hold. Furthermore, let
$\mu_{q_1 N}(x, \matXi \matXi^H)$ denote the normalized spectral
function of the square $q_1 N \times q_1 N$ matrix argument, i.e.,
the empirical eigenvalue distribution of the matrix $\matXi
\matXi^H.$ Then, for almost all $x$ with probability one,
\begin{equation*}
\lim_{N \rightarrow  \infty} |\mu_{q_1 N}(x, \matXi
\matXi^H)-F_{q_1 N} (x) |=0
\end{equation*}
where $F_{q_1 N}(x)$ is the distribution function whose Stieltjes
transform is equal to
\begin{equation}
\int_{0}^{+\infty} (x+\alpha)^{-1} \mathrm{d}F_{q_1 N} (x) = (q_1
N )^{-1} \mathrm{tr}[\boldsymbol{\widetilde{\Psi}}]^{-1}
\end{equation}
with
$\boldsymbol{\widetilde{\Psi}}=\mathrm{diag}(\boldsymbol{\Psi}_{nn})_{n=1
\dots N}$ nonnegative definite analytic matrix for
$\mathrm{Re}(\alpha)>0$ and $\boldsymbol{\Psi}_{nn}$ satisfying
the canonical system of equations (\ref{C_1girko}) and
(\ref{C_2girko}).
\end{lemma}

\begin{lemma}\cite{bai:98}\label{lemma_bai} Let $\Vec{x}=(x_1, x_2, \ldots, x_N)$
be an $N$-dimensional column vector of complex i.i.d. elements
with zero mean and unit variance and $\Mat{C}$ be an $N \times N $
complex matrix. Then, for any $p \geq 2$
\begin{equation}\label{bai_bound}
\mathrm{E}|\Vec{x}^H \Mat{C} \Vec{x}- \mathrm{tr} \Mat{C}|^p \leq
K_p \left( \left(\mathrm{E}|x_1|^4 \mathrm{tr} \Mat{CC}^H
\right)^{\frac{p}{2}} +  \left(\mathrm{E}|x_1|^{2p} \mathrm{tr}
(\Mat{CC}^H)^{\frac{p}{2}} \right) \right)
\end{equation}
with $K_p$ positive constant independent of $N.$
\end{lemma}

\section{Proof of Theorem \ref{theo:sinr_MMSE_chip_asynch}}\label{section:proof_theo_sinr_MMSE_chip_asynch}
Let us consider the $r$-block-wise circulant matrices of order $
\N $, $\boldsymbol{\Phi}_k,$ $k=1, \ldots K$
defined in Theorem \ref{theo:sinr_MMSE_chip_asynch}, and let us
denote with $\Mat{F}_{\N}^H$ the unitary Fourier transform matrix
of dimensions $\N \times \N$ with $(\ell, m)$ element given by
\begin{equation}\label{fourier_matrix}
(\Mat{F}_{\N})_{\ell,m}=\frac{1}{\sqrt{\N}}\mathrm{e}^{\frac{j 2 \pi}{\N}(\ell-1)(m-1)}.
\end{equation}
We can extend the well known results on the diagonalization of circulant
matrices\footnote{A circulant matrix $\Mat{C}=\mathfrak{F}(f(x))$ of order $N$
can be decomposed as $\Mat{C}= \Mat{F}_{\N} \Mat{D}
\Mat{F}_{\N}^H$, with $\Mat{D}=\mathrm{diag}(f(0),
f(\frac{2\pi}{\N}), \ldots , f(2 \pi\frac{ (\N-1)}{\N}) ).$ } \cite{gray:72}
to decompose the $r$-block-wise circulant matrices
$\boldsymbol{\Phi}_k,$ $k=1, \ldots K$ as
\begin{equation}\label{decomposition}
\boldsymbol{\Phi}_k=(\Mat{F}_{\N} \otimes
\Mat{I}_{r} ) \boldsymbol{\Delta}_{\phi,r}({\tau}_k)
\Mat{F}_{\N}
\end{equation}
where $\boldsymbol{\Delta}_{\phi,r}({\tau}_k)$ is an $r
\N \times \N$ block diagonal matrix with $\ell^{\text{th}}$ block
given by
\begin{equation}\label{delta_block}
(\boldsymbol{\Delta}_{\phi,r}({\tau}_k))_{\ell, \ell} = \boldsymbol{\Delta}_{\phi,r}\left(2 \pi \frac{\ell-1}{N}, {\tau}_k \right) \end{equation}
and $(\Mat{F}_{\N}\otimes
\Mat{I}_r)$ is a unitary matrix.

The matrix ${\matS}$ can then be rewritten as
\begin{equation*}
{\matS}=(\Mat{F}_{\N} \otimes \Mat{I}_r)
({\boldsymbol{\Delta}}_{\phi,r} ({\tau}_{1})
\widetilde{\vecs}_1, {\boldsymbol{\Delta}}_{\phi,r}
({\tau}_{2}) \widetilde{\vecs}_2, \ldots,
{\boldsymbol{\Delta}}_{\phi,r} ({\tau}_{\K})
\widetilde{\vecs}_{\K}),
\end{equation*} with $\widetilde{\vecs}_k=\Mat{F}_{\N}^H \vecs_k$.
Assuming the elements of the spreading sequence $\vecs_k$ i.i.d.
Gaussian distributed,  $\widetilde{\vecs}_k$ is also a vector with
i.i.d. Gaussian distributed elements having the same distribution
as the elements of $\vecs_k.$ Since the eigenvalues of any matrix
$\Mat{X}$ are invariant with respect to left multiplication by a
unitary matrix $\Mat{U}$ and right multiplication by
$\Mat{U}^{H}$, i.e., the eigenvalues of the matrix $\Mat{X}$
coincides with the eigenvalues of the matrix $\Mat{UXU}^H$ , then
the singular values of the matrices ${\matS}$ and
$\widetilde{\matS}=({\boldsymbol{\Delta}}_{\phi,r} ({\tau}_{1})
\widetilde{\vecs}_1, {\boldsymbol{\Delta}}_{\phi,r}
({\tau}_{2}) \widetilde{\vecs}_2, \ldots,
{\boldsymbol{\Delta}}_{\phi,r} ({\tau}_{\K})
\widetilde{\vecs}_{\K})$ coincide. The same properties holds for the
matrices ${\matH}$ and $\widetilde{\matH}=\widetilde{\matS}
\matA.$ It is straightforward to verify that also
$\mathrm{SINR}_{k}$ is invariant with respect to such a transform.
In fact,
\begin{align}
\mathrm{SINR}_k &= {\vech}_k^H\left({\matH}_k
{\matH}_k^H + \sigma^2 \Mat{I} \right)^{-1}
{\vech}_k \nonumber \\
&= |a_{k}|^2 \widetilde{\vecs}_k^H\left(\widetilde{\matH}_k
\widetilde{\matH}_k^H + \sigma^2 \Mat{I} \right)^{-1}
\widetilde{\vecs}_k \nonumber
\end{align}
with ${\matH}_k$ and $\widetilde{\matH}_k$ obtained from
the matrices ${\matH}$ and $\widetilde{\matH},$
respectively, by suppressing the $k^{\mathrm{th}}$ column.
Therefore, in the following we focus on the analysis of the system
with transfer matrix $\widetilde{\matH}.$

The matrix $\widetilde{\matH}$ is a matrix structured in blocks of
dimensions $r \times 1.$ The block $(n,k)$
$\widetilde{\vech}_{n,k},$ $n=1, \ldots N$ and $k=1,\ldots K,$ is
given by
\begin{equation*}
\widetilde{\vech}_{n,k}= |a_{k}|^2
(\boldsymbol{\Delta}_{\phi,r}({\tau}_k))_{nn}
\widetilde{s}_{n,k}
\end{equation*}
where $\widetilde{s}_{n,k}$ is a Gaussian random variable with
zero mean and variance
$\E\{|\widetilde{s}_{n,k}|^2\}=\frac{1}{N}.$ Additionally, the
variables $\widetilde{s}_{n,k}$ are i.i.d.. Therefore, conditions
H-1 and H-2 for the applicability of Lemma \ref{lemma_T_R} and
Lemma \ref{lemma_convergence_canonical_eqns} are satisfied.
Condition H-3 of Lemma \ref{lemma_T_R} is satisfied. In fact,
\begin{align} \zeta & = \sup_{\N}  \left[ \max_{n=1,
\ldots \N}   \sum_{k=1}^{\K} \mathrm{E} \{\|\widetilde{\vech}_{nk}
\|^2 \} + \max_{k=1, \ldots \K}\sum_{n=1}^{\N} \mathrm{E}
\{\|\widetilde{\vech}_{nk} \|^2 \} \right] \nonumber \\  &  \leq
\sup_{\N} \left[\max_{n=1...N} \sum_{k=1}^{K}
\frac{|a_{k}|^2}{\N}
\|(\boldsymbol{\Delta}_{\phi,r}(\tilde{\tau}_k))_{nn}\|^2 +
\max_{k=1...K} \frac{|a_{k}|^2}{\N} \sum_{n=1}^{N}
\|(\boldsymbol{\Delta}_{\phi,r}({\tau}_k))_{nn}\|^2
\right]. \nonumber
\end{align}
Since the function $\Phi(\omega)$ is bounded in absolute value
with finite support also $|\Phi(\Omega, \tau)|$ is upper bounded for
any $\Omega$ and $\tau.$ Then, there exists a constant
$C_{\mathrm{MAX}}>0$ that satisfies
$\|(\boldsymbol{\Delta}_{\phi,r}({\tau}_k))_{nn} \|^2<
C_{\mathrm{MAX}}$ for any $k$ and $n.$ Additionally, the elements
$|a_{k}|$ are uniformly bounded for any $k$, i.e., $ \exists
a_{\mathrm{MAX}}>0$ such that $|a_{k}|^2 \leq a_{\mathrm{MAX}}^2$
for all $k.$ Then,
\begin{equation}\label{verifica_H_1}
\zeta \leq \sup_{K=\beta N} a_{\mathrm{MAX}}^2 C_{\mathrm{MAX}}
\left( \frac{K}{N} +1 \right) < + \infty.
\end{equation}

In order to verify the Lindeberg condition H-4 we focus on the
limit
\begin{equation*}
\eta = \lim_{K=\beta N \rightarrow \infty} \max_{N} \sum_{k=1}^{K}
\E\left( \| \widetilde{\vech}_{nk} \|^2 \chi\left( \|
\widetilde{\vech}_{nk} \| > \delta \right) \right)
\end{equation*}
for any $\delta > 0.$ Let us observe that $\forall n, k $
\begin{align}
\E\left( \|\widetilde{\vech}_{nk} \|^2 \chi(\|
\widetilde{\vech}_{nk}\|
> \delta) \right) &= |a_{k}|^2  \| (\boldsymbol{\Delta}_{\phi,r}({\tau}_k))_{nn} \|^2 \int_{\left\{|\widetilde{s}_{nk}|^2>
\frac{\delta^2}{|a_{k}|^2
\|\boldsymbol{\Delta}_{\phi,r}({\tau}_k) \|^2} \right\}}
|\widetilde{s}_{nk}|^2
\mathrm{d}F(\widetilde{s}_{nk}) \nonumber \\
& \leq \frac{|a_{k}|^4 \|
(\boldsymbol{\Delta}_{\phi,r}({\tau}_k))_{nn} \|^4}{
\delta^2} \int_{\left\{|\widetilde{s}_{nk}| \geq 0 \right\}} |\widetilde{s}_{nk}|^{4}
\mathrm{d}F(\widetilde{s}_{nk}) \nonumber
\end{align}
where $F(\widetilde{s}_{nk})$ is the cumulative distribution function of $\widetilde{s}_{nk}.$ By using the fact that $\widetilde{s}_{nk}$ is a complex Gaussian variable
with variance $\E\{|s_{nk}|^2\}=\frac{1}{N}$ and forth moment
$\E\{|s_{nk}|^4\}=\frac{2}{N^2}$ and the bounds on $|a_{k}|^2$
and $ \| (\boldsymbol{\Delta}_{\phi,r}({\tau}_k))_{nn}
\|,$ it holds
\begin{align}
\max_{n=1, \ldots N} \E\left( \|\widetilde{\vech}_{nk} \|^2 \chi(\|
\widetilde{\vech}_{nk}\|
> \delta) \right) \leq \frac{2 a_{\mathrm{MAX}}^4 C_{\mathrm{MAX}}^4}{\delta^2
N^2}.
\end{align}

Then, $\eta=0$ since
\begin{equation*}
0 \leq \eta \leq   \frac{2 a_{\mathrm{MAX}}^4
C_{\mathrm{MAX}}^4}{\delta^2 } \lim_{K=\beta N \rightarrow \infty}
\sum_{k=1}^{K} \frac{1}{N^2} =0.
\end{equation*}
Similarly, it can be shown that
\begin{equation*}
\lim_{K=\beta N \rightarrow \infty} \max_{k=1, \ldots, K}
\sum_{n=1}^{N} \E\left( \|\widetilde{\vech}_{nk} \|^2
\chi(\|\widetilde{\vech}_{nk} \|^2 >\delta) \right)=0
\end{equation*}
and the Lindeberg condition H-4 is satisfied.

From Lemma \ref{lemma_T_R},  $\Mat{U}_{p \ell}(\alpha),$ $p,
\ell=1,\ldots, N,$ the blocks of the matrix
$\Mat{U}(\alpha)=(\widetilde{\matH}_k \widetilde{\matH}_k^H + \alpha
\Mat{I} )^{-1}$ converge in the first mean to $r \times r$
matrices $\Mat{V}_{p \ell}=(\Mat{C}^{(1)}_{\ell \ell})^{-1}
\delta_{p \ell},$ $p, \ell = 1, \ldots N,$ and
$\Mat{C}^{(1)}_{\ell \ell}$ defined similarly as in Lemma
\ref{lemma_T_R}. Additionally, from Lemma
\ref{lemma_convergence_canonical_eqns} the matrices $\Mat{C}_{\ell
\ell}^{(1)}$ can be obtained as solution of the canonical system
of equations (\ref{C_1girko}) and (\ref{C_2girko}) asymptotically
as $K= \beta N \rightarrow \infty.$ Equations (\ref{C_1girko}) can
be rewritten as
\begin{align} \boldsymbol{\Upsilon}_{nn}^{(1)} & = \alpha \Mat{I}_{r} +
\sum_{k=1}^{\K} \mathrm{E}\left\{
\widetilde{\vech}_{nk}[{\boldsymbol{\Upsilon}}_{kk}^{(2)}]^{-1} \widetilde{\vech}_{nk}^H \right\} \nonumber \\
& = \alpha \Mat{I}_{r} + \frac{1}{N}\sum_{k=1}^{\K}
[{\boldsymbol{\Upsilon}}_{kk}^{(2)}]^{-1} |a_{k}|^2
\boldsymbol{\Delta}_{\phi,r}\left(2 \pi \frac{n-1}{N},
{\tau}_k \right) \boldsymbol{\Delta}_{\phi,r}^H \left(
2 \pi \frac{n-1}{N}, {\tau}_k \right)  \qquad n=1, \ldots N
\label{C_1_C}
\end{align}
with $\boldsymbol{\Delta}_{\phi,r} \left( x, \tau \right)$ defined
in (\ref{delta_phi_r_def}) and taking into account  that $(\boldsymbol{\Delta}_{\phi,r}
\left({\tau}_k \right))_{nn}=\boldsymbol{\Delta}_{\phi,r} \left( 2 \pi \frac{n-1}{N},
{\tau}_k \right)$ in
(\ref{C_1_C}). Equations (\ref{C_2girko}) specialize
to
\begin{align}
\boldsymbol{\Upsilon}_{kk}^{(2)} & =  1 + \sum_{n=1}^{\N}
\mathrm{E}\left\{
\widetilde{\vech}_{nk}^H[{\boldsymbol{\Upsilon}}_{nn}^{(1)}]^{-1} \widetilde{\vech}_{nk} \right\} \nonumber \\
& = 1 + \frac{|a_{k}|^2}{N} \sum_{n=1}^{\N}
\boldsymbol{\Delta}_{\phi,r}^H\left( 2 \pi \frac{n-1}{N},
{\tau}_k \right)
[{\boldsymbol{\Upsilon}}_{nn}^{(1)}]^{-1}
\boldsymbol{\Delta}_{\phi,r} \left( 2 \pi \frac{n-1}{N},
{\tau}_k \right)   \qquad k=1, \ldots K \label{C_2_B}.
\end{align}  By substituting (\ref{C_2_B}) in
 (\ref{C_1_C}) and considering the canonical system of equations
 as $K=\beta N \rightarrow \infty$ we obtain
\begin{align}
\boldsymbol{\Upsilon}^{(1)} (\Omega)&=\alpha \Mat{I}_{r}+ \beta
\int_{\mathcal{S}} \frac{\lambda
\boldsymbol{\Delta}_{\phi,r}\left( \Omega, \tau
\right)\boldsymbol{\Delta}_{\phi,r}^{H}\left( \Omega, \tau \right)  f_{|\matA|^2,T}(\lambda, \tau) \mathrm{d}\, \lambda \mathrm{d}\,\tau}{1+ \frac{\lambda}{2 \pi}
\int_{\mathcal{X}} \boldsymbol{\Delta}_{\phi,r}^{H}\left( \Omega^{\prime}, \tau
\right) [\boldsymbol{\Upsilon}^{(1)}(\Omega^{\prime})]^{-1}
\boldsymbol{\Delta}_{\phi,r}\left( \Omega^{\prime}, \tau \right) \mathrm{d}\, \Omega^{\prime}}
\nonumber
\end{align}
with $\Omega \in [0, 2 \pi],$ or since $\boldsymbol{\Delta}_{\phi,r}(\Omega,
\tau)$ is periodical in $\Omega$ with  period $2 \pi$, $\Omega$ can equivalently
varies in the interval $\mathcal{X}=\left[ -\pi,
\pi \right]$. Here, $\mathcal{S}$ denotes the support of
the distribution function $F_{|\matA|^2,T}(\lambda, \tau).$ By
defining
$\boldsymbol{\Upsilon}(\Omega)=[\boldsymbol{\Upsilon}^{(1)}(\Omega)]^{-1}$
we obtain (\ref{multiuser_efficiency_mat}). It follows from Lemma
\ref{lemma_convergence_canonical_eqns}
\begin{equation*}
\lim_{K=\beta N \rightarrow + \infty} \Mat{C}_{nn}^{(1)} =
\boldsymbol{\Upsilon}^{-1}\left(2 \pi \frac{n}{N}\right).
\end{equation*}

The convergence in the first mean and thus in probability of
$\mathrm{SINR}_k=\widetilde{\vech}_k^H \Mat{U}(\sigma^2)
\widetilde{\vech}_k $ to the quantity $\varrho=\frac{|a_{k}|^2}{2 \pi} \int
\boldsymbol{\Delta}_{\phi,r}^H(\Omega, {\tau}_k)
\Mat{\Upsilon}(\Omega) \boldsymbol{\Delta}_{\phi,r}(\Omega,
{\tau}_k) \mathrm{d}\Omega$ is proven if $\eta_1=\E
\left|\widetilde{\vech}_k^H \Mat{U}(\sigma^2) \widetilde{\vech}_k -
\varrho \right| $ vanishes asymptotically, i.e.,
\begin{equation}\label{limit_convergence_probability}
\lim_{\begin{subarray}{c}
  \K,\N \rightarrow \infty\\
  \frac{\K}{\N} \rightarrow \beta
\end{subarray}} \eta_1=0.
\end{equation}

The rest of the proof is focused on showing
(\ref{limit_convergence_probability}). Let us observe
\begin{align}
\eta_1 &\leq \E \left|\widetilde{\vech}_k^H \Mat{U}(\sigma^2)
\widetilde{\vech}_k - \widetilde{\vech}_k^H \Mat{V} \widetilde{\vech}_k
\right|+ \E |\widetilde{\vech}_k^H \Mat{V} \widetilde{\vech}_k-
\varrho| \label{aux_riferimento}
\end{align}
where the triangular inequality\footnote{Given two matrices
$\Mat{A}$ and $\Mat{B}$ with consistent dimensions the following
inequalities hold: \begin{align} |\Mat{A} \Mat{B}|\leq |\Mat{A}|
|\Mat{B}| \qquad & \text{Submultiplicative inequality of spectral
norms; } \nonumber \\
|\Mat{A}+ \Mat{B}|\leq |\Mat{A}| + |\Mat{B}| \qquad &
\text{Triangular inequality of spectral norms. } \nonumber
\end{align}} of the spectral norm is applied and
$\Mat{V}=\mathrm{diag}([\matC^{(1)}_{kk}(\sigma^2)]^{-1})_{k=1,
\ldots,N}$ is defined in Lemma
\ref{lemma_convergence_canonical_eqns}.

By applying the submultiplicative inequality for spectral norms
and the triangular inequality  to the first term of
(\ref{aux_riferimento}) we obtain

\begin{align}
\E |\widetilde{\vech}_k^H (\Mat{U}(\sigma^2)-\Mat{V})
\widetilde{\vech}_k| & = \E \left|\sum_{i, \ell}
\widetilde{s}_{ik}^{*}
\boldsymbol{\Delta}_{\phi,r}^H\left(2 \pi \frac{i-1}{N},
{\tau}_k \right) (\Mat{U}(\sigma^2)-\Mat{V})_{i \ell}
\boldsymbol{\Delta}_{\phi,r}\left(2 \pi \frac{\ell-1}{N},
{\tau}_k \right) \widetilde{s}_{\ell k}\right| \nonumber \\
& \leq \sum_{i, \ell} \E (|(\Mat{U}(\sigma^2))_{i \ell} -
\Mat{V}_{i \ell}|) \E |\widetilde{s}_{ik}^{*} \widetilde{s}_{\ell
k}| \boldsymbol{\Delta}_{\phi,r}^H\left(2 \pi \frac{i-1}{N}, {\tau}_k \right)
\boldsymbol{\Delta}_{\phi,r}\left(2 \pi \frac{\ell-1}{N},{\tau}_k \right)  \nonumber \\
& = \sum_{i} \E |(\Mat{U}(\sigma^2))_{ii}-\Mat{V}_{ii}|
\frac{\left\|\boldsymbol{\Delta}_{\phi,r}\left(2 \pi \frac{i-1}{N},
{\tau}_k \right)\right\|^2}{N} \nonumber \\
& \leq   \sum_{i} \frac{\|\boldsymbol{\Delta}_{\phi,r}\left(2 \pi \frac{i-1}{N},
{\tau}_k \right)\|^2}{N}  \max_{i} \E
|(\Mat{U}(\sigma^2))_{ii}-\Mat{V}_{ii}|. \nonumber
\end{align}
Thanks to Lemma \ref{lemma_T_R} and the fact that
$\|\boldsymbol{\Delta}_{\phi,r}\left( 2 \pi \frac{i-1}{N},
{\tau}_k \right)\|^2 \leq K_{\mathrm{MAX}}$ for all
$i=1, \ldots N$ and ${\tau}_k$

\begin{equation*}
\lim_{\begin{subarray}{c}
  \K, \N \rightarrow \infty \\
  \frac{\K}{\N} \rightarrow \beta
\end{subarray}} \E |\widetilde{\vech}_k^H (\Mat{U}(\sigma^2)-\Mat{V})
\widetilde{\vech}_k| = 0.
\end{equation*}
In order to prove the convergence to zero of $\eta_2 = \E
|\widetilde{\vech}_k^H \Mat{V} \widetilde{\vech}_k - \varrho|$ we
consider {\small \begin{align}
\eta_2^2 & \leq  \E |\widetilde{\vech}_k^H \Mat{V} \widetilde{\vech}_k - \varrho |^2 \nonumber \\
& =  \E((\widetilde{\vech}_k^H \Mat{V} \widetilde{\vech}_k)^2 -2
\varrho
\widetilde{\vech}_k^H \Mat{V} \widetilde{\vech}_k  + \varrho^2) \nonumber  \\
& =  \E \left( |a_{k}|^4\sum_{ij}
\boldsymbol{\Delta}_{\phi,r}^H\left(2 \pi\frac{i-1}{N},
{\tau}_k \right) \Mat{V}_{ii}
\boldsymbol{\Delta}_{\phi,r}\left(2 \pi\frac{i-1}{N},
{\tau}_k \right)
\boldsymbol{\Delta}_{\phi,r}^H\left(2 \pi \frac{j-1}{N},
{\tau}_k \right) \Mat{V}_{jj}
\boldsymbol{\Delta}_{\phi,r}\left(2 \pi \frac{j-1}{N},
{\tau}_k \right) |\widetilde{s}_{ik}|^2
|\widetilde{s}_{jk}|^2 \right. \nonumber
\\& \left. \quad -2 \varrho |a_{k}|^2 \sum_{i}
\boldsymbol{\Delta}_{\phi,r}^H\left(2 \pi \frac{i-1}{N},
{\tau}_k \right) \Mat{V}_{ii}\boldsymbol{\Delta}_{\phi,r}\left(2 \pi \frac{i-1}{N},
{\tau}_k \right) |\widetilde{s}_{ik}|^2 + \varrho^2
\right)
\label{converg_step_3} \\
&= \frac{2 |a_{k}|^4}{\N^2}\sum_i
\left(\boldsymbol{\Delta}_{\phi,r}^H\left(2 \pi \frac{i-1}{N},
{\tau}_k \right) \Mat{V}_{ii} \boldsymbol{\Delta}_{\phi,r}^H\left(2 \pi \frac{i-1}{N},
{\tau}_k \right)\right)^2 + \frac{|a_{k}|^4}{\N^2}
\sum_{\begin{subarray}{c}
  i,j \\
  i \neq j
\end{subarray}} \boldsymbol{\Delta}_{\phi,r}^H\left(2 \pi \frac{i-1}{N},
{\tau}_k \right) \Mat{V}_{ii} \boldsymbol{\Delta}_{\phi,r}^H\left(2 \pi \frac{i-1}{N},
{\tau}_k \right)\nonumber \\
& \quad \times  \boldsymbol{\Delta}_{\phi,r}^H\left(2 \pi \frac{j-1}{N},
{\tau}_k \right) \Mat{V}_{jj} \boldsymbol{\Delta}_{\phi,r}^H\left(2 \pi \frac{j-1}{N},
{\tau}_k \right) -\frac{2 \varrho}{N} \sum_{i}
\boldsymbol{\Delta}_{\phi,r}^H\left(2 \pi \frac{i-1}{N},
{\tau}_k \right) \Mat{V}_{ii} \boldsymbol{\Delta}_{\phi,r}\left(2 \pi \frac{i-1}{N},
{\tau}_k \right)+ \varrho^2. \label{converg_step_4}
\end{align}}
From (\ref{converg_step_3}) to (\ref{converg_step_4}) we make use
of the assumption that  $\widetilde{s}_{ij}$ is a complex Gaussian
variable circularly invariant with variance $N^{-1}.$ Let us
observe that the spectral norm of $\Mat{\Upsilon}(\Omega)$ and
$\Mat{V}_{ii},$ for any $i,$ are bounded by
$|\Mat{\Upsilon}(\Omega)|<\sigma^2$ and $|\Mat{V}_{ii}|< \sigma^2.$
Then, the first term in (\ref{converg_step_4}) vanishes as $\N
\rightarrow \infty.$ By appealing Lemma
\ref{lemma_convergence_canonical_eqns}, for any $i$, $\Mat{V}_{ii}
\rightarrow \Mat{\Upsilon}\left(2 \pi \frac{i}{N}\right)$ as $\K, \N
\rightarrow \infty$ with $\frac{\K}{\N} \rightarrow \beta.$ Then,
the second and third terms in (\ref{converg_step_4}) converge to
$\varrho^2$ and $-2\varrho^2,$ respectively. We can conclude that
\begin{equation*}
\lim_{\begin{subarray}{c}
  \K,\N \rightarrow \infty \\
  \frac{\K}{\N} \rightarrow \beta
\end{subarray}} \eta_2^2 =0
\end{equation*}
and $\eta_2 \rightarrow 0$ as $\K, \N \rightarrow \infty$ as
$\frac{\K}{\N} \rightarrow \beta.$ Therefore,
(\ref{limit_convergence_probability}) and thus the convergence in
the first mean of $\mathrm{SINR}_{k}$ is proven. The Markov
inequality implies that, $\forall \varepsilon>0 $
\begin{equation*}
\lim_{\begin{subarray}{c}
  \K,\N \rightarrow \infty \\
  \frac{\K}{\N} \rightarrow \beta \end{subarray}} \mathrm{Pr}\{|\widetilde{\vech}_k^H \Mat{U}(\sigma^2) \widetilde{\vech}_k - \varrho| > \varepsilon
  \}  \leq \frac{1}{\varepsilon} \lim_{\begin{subarray}{c}
  \K,\N \rightarrow \infty \\
  \frac{\K}{\N} \rightarrow \beta \end{subarray}} \E|\widetilde{\vech}_k^H \Mat{U}(\sigma^2) \widetilde{\vech}_k - \varrho|=0 \nonumber
\end{equation*}
and the convergence in probability stated in Theorem
\ref{theo:sinr_MMSE_chip_asynch} is proven.

This concludes the proof of Theorem
\ref{theo:sinr_MMSE_chip_asynch}.

\section{Proof of Corollary \ref{cor_MMSE_raised_cosine}}\label{section:proof_cor_MMSE_raised_cosine}
In order to prove Corollary \ref{cor_MMSE_raised_cosine} we
rewrite the limit $\mathrm{SINR}_k$ in
(\ref{limit_SINR_MMSE_genaral}) as
\begin{equation}\label{limit_SINR_MMSE_modified}
\lim_{K,N \rightarrow \infty} \mathrm{SINR}_k = \frac{{|a_{k}|^2}}{2 \pi}
\int_{-\pi}^{\pi}
\mathrm{tr}\left(\Mat{\Upsilon}(\Omega) \Mat{Q}(\Omega, {\tau}_k)
\right) \mathrm{d}\Omega
\end{equation}
and the fixed point equation (\ref{multiuser_efficiency_mat}) as
\begin{equation}\label{multiuser_efficiency_modified}
\Mat{\Upsilon}^{-1}(\Omega) = \sigma^2 \Mat{I}_{r} + \beta
\int_{0}^{+\infty} \int_{0}^{T_c} \dfrac{\lambda \Mat{Q}(\Omega, \tau)
f_{|\matA|^2, T }(\lambda, \tau) \mathrm{d}\lambda \mathrm{d} \tau }{1 + \frac{\lambda}{2 \pi}
\int_{-\pi}^{\pi} \mathrm{tr}\left(
\Mat{\Upsilon}(\Omega^{\prime})\Mat{Q}(\Omega^{\prime}, \tau)\right) \mathrm{d}\Omega^{\prime} } \qquad
-\pi \leq \Omega \leq \pi
\end{equation}
with $\Mat{Q}(\Omega, \tau)$ defined in (\ref{matrixQ(xt)_def}). The
matrix $\Mat{Q}(\Omega, \tau)$ can be decomposed as in
(\ref{Q(xt)_decomposition}). Thanks to the assumptions on $\Phi(\Omega )$ in Corollary \ref{cor_MMSE_raised_cosine}, the
conditions on $\Mat{Q}(\Omega)$ and $\overline{\Mat{Q}}(\Omega,\tau)$ in
Lemma \ref{Lemma:traceBQ} and Lemma \ref{lemma:invarianceQT} are
satisfied. First we show that $\Mat{\Upsilon}(\Omega),$ the unique
solution of (\ref{multiuser_efficiency_modified}) in the class of
nonnegative definite analytic functions in
$\mathrm{Re}(\sigma^2)>0,$ is an $r \times r $ matrix with
eigenbasis $\Mat{U}(\Omega)$ defined in (\ref{U_definition}). Let us
assume that $\Mat{\Upsilon}(\Omega)= \Mat{U}(\Omega)
\widetilde{\Mat{\Upsilon}}(\Omega) \Mat{U}^H(\Omega) $ with elements of
$\widetilde{\Mat{\Upsilon}}(\Omega)$ nonnegative for all $\Omega \in \left[
-\pi, \pi \right].$ By appealing to Lemma
\ref{lemma:matrix_with_eigenbasis_U} $\Mat{\Upsilon}(\Omega)$ is of
form (\ref{matrixBform}). Then, by applying Lemma
\ref{Lemma:traceBQ} it results $\mathrm{tr}\left(\Mat{\Upsilon}(\Omega)
\overline{\Mat{Q}}(\Omega, \tau ) \right)=0$ for all $\Omega \in \left[
-\pi, \pi \right].$ Therefore,
\begin{align}
\int_{-\pi}^{\pi} \mathrm{tr}\left(\Mat{\Upsilon}
(\Omega) \Mat{Q}(\Omega,\tau) \right) \mathrm{d}\Omega &=
\int_{-\pi}^{\pi} \mathrm{tr}\left(\Mat{\Upsilon}
(\Omega^{\prime}) \Mat{Q}(\Omega^{\prime}) \right) \mathrm{d}\Omega^{\prime}  \nonumber \\
&= \int_{-\pi}^{\pi} \mathrm{tr}\left(\widetilde{\Mat{\Upsilon}} (\Omega^{\prime}) \Mat{D}(\Omega^{\prime}) \right) \mathrm{d}\Omega^{\prime} \geq 0 \nonumber
\end{align}
with $\Mat{D}(\Omega)$ defined as in Lemma
\ref{lemma:Q(x)_decomposition}. Let us notice that
$\int_{0}^{T_c}\overline{\Mat{Q}}(\Omega, \tau)
\mathrm{d}F_{T}(\tau)=0$ for all $\Omega.$ Thanks to this property, the
assumption of independence of the random variables $\lambda$ and
$\tau,$ and to the uniform distribution of $\tau,$
(\ref{multiuser_efficiency_modified}) can be rewritten as
\begin{equation}\label{multiuser_efficiency_modified2}
\widetilde{\Mat{\Upsilon}}^{-1}(\Omega) = \sigma^2 \Mat{I}_{r} + \beta
\left(\int_{0}^{+\infty}  \dfrac{\lambda \mathrm{d} F_{|\matA|^2 }
(\lambda) }{1 + \frac{\lambda}{2 \pi} \int_{-\pi}^{\pi}
\mathrm{tr}\left( \widetilde{\Mat{\Upsilon}}(\Omega^{\prime})\Mat{D}(\Omega^{\prime})\right)
\mathrm{d}\Omega^{\prime} } \right) \qquad -\pi \leq \Omega \leq \pi.
\end{equation}
Since $\widetilde{\Mat{\Upsilon}}(\Omega)$  and $\Mat{D}$ are diagonal
matrices, the matrix equation (\ref{multiuser_efficiency_modified})
reduces to a system of $L$ scalar equations. Furthermore, all the
quantities that appears in the right hand side of the system of
equations (\ref{multiuser_efficiency_modified2}) are nonnegative
under the assumption that $\widetilde{\Mat{\Upsilon}}(\Omega)$ is a
nonnegative definite matrix and
(\ref{multiuser_efficiency_modified2}) admits a nonnegative
definite solution for $\mathrm{Re}(\sigma^2)>0.$ The existence of
a nonnegative definite solution of the system of equations
(\ref{multiuser_efficiency_modified2}) implies also a solution of
the fixed point matrix equation
(\ref{multiuser_efficiency_modified}) given by $\Mat{\Upsilon}(\Omega)=
\Mat{U} (\Omega) \widetilde{\Mat{\Upsilon}}(\Omega) \Mat{U}^H(\Omega) . $ Let
$\widetilde{\upsilon}_s(\Omega)$ be the $s^{\mathrm{th}}$ diagonal
element of $\widetilde{\Mat{\Upsilon}}(\Omega)$ and let us recall that
the $s^{\mathrm{th}}$ diagonal element of $\Mat{D}(\Omega)$ is given in
(\ref{elements_D(x)}). Then,
(\ref{multiuser_efficiency_modified2}) reduces to
\begin{multline}\label{multiuser_efficiency_modified3}
\widetilde{\upsilon}_s^{-1}(\Omega) = \sigma^2 + \beta \frac{r}{T_c^2}
\left|\Phi\left( \frac{\Omega}{T_c} -\mathrm{sign}(\Omega)\frac{2 \pi}{T_c}
\left(\left\lfloor  \frac{r-1}{2} \right\rfloor -s+1 \right)
\right)  \right|^2 \\ \times \int_{0}^{+\infty}
\dfrac{\lambda \mathrm{d} F_{|\matA|^2 } (\lambda) }{1 +
\frac{\lambda r}{2 \pi T_c^2} \int_{-\pi}^{\pi}
\sum_{\ell=1}^{r} \widetilde{\upsilon}_{\ell}( \Omega^{\prime}) \left|\Phi\left(
\frac{\Omega^{\prime}}{T_c} -\frac{2 \pi}{T_c}\mathrm{sign}(\Omega^{\prime}) \left(\left \lfloor
\frac{r-1}{2} \right \rfloor -\ell+1 \right) \right)
\right|^2    \mathrm{d}\Omega^{\prime} } \\ \qquad -\pi \leq \Omega \leq
\pi \text{  and } s=1, \ldots r .
\end{multline}
By changing the variable $y=\Omega-\mathrm{sign}(\Omega)2 \pi \left(\left \lfloor
\frac{r-1}{2}  \right \rfloor -s +1\right)$ and defining the
function $\upsilon(y)$ in the interval $\left(-r \pi,
r \pi \right)$ as follows
\begin{equation}
\upsilon(y)=
  \begin{cases}
    \widetilde{\upsilon}_{s}\left( y-2 \pi \left(\left\lfloor \frac{r-1}{2} \right\rfloor -s+1 \right) \right) & \left\lfloor \frac{r-1}{2} \right\rfloor -s+\frac{1}{2} \leq \frac{y}{2 \pi} \leq \left\lfloor \frac{r-1}{2} \right\rfloor -s+1 , \\
  \widetilde{\upsilon}_{s}\left( y+2 \pi \left(\left\lfloor \frac{r-1}{2} \right\rfloor -s+1 \right)\right) & s-1-\left\lfloor \frac{r-1}{2} \right\rfloor  \leq \frac{y}{2 \pi} \leq s-\frac{1}{2} -\left\lfloor \frac{r-1}{2} \right\rfloor
  \end{cases} \qquad s=1, \ldots,r
\end{equation}
the system of equations (\ref{multiuser_efficiency_modified3}) can
be rewritten as
\begin{multline}\label{fp_A}
\upsilon^{-1}(\Omega) = \sigma^2 + \beta \frac{r}{T_c^2}
\left|\Phi\left( \frac{y}{T_c} \right) \right|^2
\int_{0}^{+\infty} \dfrac{\lambda \mathrm{d} F_{|\matA|^2 }
(\lambda) }{1 + \frac{\lambda r}{2 \pi T_c}
\int_{-r \pi}^{r \pi} \upsilon(t) \left|\Phi\left(
\frac{t}{T_c} \right) \right|^2    \mathrm{d}t } \qquad
|y| \leq r \pi .
\end{multline}
A similar approach applied to (\ref{limit_SINR_MMSE_modified})
yields
\begin{align}
\lim_{K,N \rightarrow \infty} \mathrm{SINR}_k &= \frac{|a_{k}|^2
r}{T_c^2} \sum_{s=1}^{r} \int_{-\pi}^{\pi} \left|
\Phi\left( \Omega- 2 \pi \mathrm{sign}(\Omega)\left( \left \lfloor \frac{r-1}{2}
\right\rfloor -s+1 \right) \right) \right|^2
{\upsilon}(\Omega) \mathrm{d}\Omega \nonumber \\
&=\frac{|a_{k}|^2 r}{T_c^2} \int_{-\pi r}^{\pi r}
\left| \Phi\left(\frac{\Omega}{T_c} \right) \right|^2 \upsilon(\Omega)
\mathrm{d}\Omega. \label{limit_B}
\end{align}
Let us recall that the variance of the discrete white noise is $\sigma^2= \frac{N_0 r}{T_c}.$  Additionally, let us define the function
\begin{equation}\label{efficiency_spectrum}
    \eta\left( \frac{\Omega}{T_c} \right) = \frac{r N_0}{T_c E_{\phi}} \left| \Phi\left( \frac{\Omega}{T_c}\right) \right|^2 \upsilon(\Omega).
\end{equation}
By substituting (\ref{defMUE}) in (\ref{fp_A}) and (\ref{limit_B}), using definition (\ref{efficiency_spectrum}), and $\omega=\frac{\Omega}{T_c}$  we obtain the fixed point equation
(\ref{fixed_point_eq_cor_MMSE_raised_cos}) and the limit
(\ref{SINR_cor_MMSE_raised_cos}), respectively. This concludes the proof of Corollary
\ref{cor_MMSE_raised_cosine}.

\section{Proof of Theorem \ref{theo:small_bandwidth_MMSE}}\label{section:proof_theo_small_bandwidth_MMSE}
The proof of Theorem \ref{theo:small_bandwidth_MMSE} follows along
the line of the proof of Theorem \ref{theo:sinr_MMSE_chip_asynch}.
In this case $\Mat{\Delta}_{\phi,r} (\Omega, \tau)= \frac{\mathrm{e}^{j
\frac{\sqrt{r}}{T_c} \tau \Omega }}{T_c} \Phi^{*}\left( \frac{\Omega }{T_c}\right) \Vec{e}(\Omega)$ and the matrix $\Mat{Q}(\Omega, \tau)$
is independent of $\tau.$ Specifically,
\begin{equation*}
\Mat{Q}(\Omega, \tau)= \frac{r}{T_c^2} \left|\Phi\left(  \frac{\Omega}{T_c}\right)\right|^2 \Vec{e}(\Omega) \Vec{e}^H(\Omega).
\end{equation*}
Then, applying the same approach as in Theorem
\ref{theo:sinr_MMSE_chip_asynch} Lemma \ref{lemma_T_R} and Lemma
\ref{lemma_convergence_canonical_eqns} yield
\begin{equation*}
\lim_{K, N \rightarrow \infty} \mathrm{SINR}_k = \frac{|a_{k}|^2
r}{2 \pi T_c^2} \int_{-\pi}^{\pi}  \left| \Phi\left(
\frac{\Omega }{T_c} \right)\right|^2 \Vec{e}^H(\Omega)
\widetilde{\Mat{\Upsilon}}( \sigma^2,\Omega) \Vec{e}(\Omega) \mathrm{d}\Omega
\end{equation*}
with
\begin{align}
\widetilde{\Mat{\Upsilon}}^{-1}( \sigma^2,\Omega)&= \sigma^2
\Mat{I}_{r}+\beta \frac{r}{T_c^2} \left| \Phi\left( \frac{\Omega}{T_c}\right)\right|^2 \Vec{e}(\Omega) \Vec{e}^H(\Omega) \int_{0}^{+\infty}
\frac{\lambda \mathrm{d}F_{|A|^2}(\lambda)}{1+ \frac{\lambda}{2 \pi}
\int_{-\pi}^{\pi} \frac{r}{T_c^2} \left|
\Phi\left(  \frac{\Omega^{\prime} }{T_c}\right)\right|^2 \Vec{e}^H(\Omega^{\prime})
\widetilde{\Mat{\Upsilon}}( \sigma^2,\Omega^{\prime}) \Vec{e}(\Omega^{\prime}) \mathrm{d}\Omega^{\prime}}
\nonumber \\
& =\sigma^2 \Mat{I}_r +\beta \int_{0}^{+\infty} \frac{\Mat{U}(\Omega)
\Mat{D}^{\prime}(\Omega)  \Mat{U}^H(\Omega) \lambda
\mathrm{d}F_{|\Mat{A}|^2}(\lambda)}{1+ \frac{\lambda}{2 \pi}
\int_{-\pi}^{\pi} \frac{r}{T_c^2} \left|
\Phi\left(  \frac{\Omega^{\prime} }{T_c} \right)\right|^2 \Vec{e}^H(\Omega^{\prime})
\widetilde{\Mat{\Upsilon}}( \sigma^2, \Omega^{\prime}) \Vec{e}(\Omega^{\prime}) \mathrm{d}\Omega^{\prime}}
\label{fixed_point_intermediate_step}
\end{align}
with $\Mat{U}(\Omega)$ defined in (\ref{U_definition}) and
$\Mat{D}^{\prime}(\Omega)$ diagonal matrix with all zero elements
except the $\left(\left\lfloor  \frac{r-1}{2} \right \rfloor +1
\right)^{\mathrm{th}}$ element, corresponding to the eigenvector
$\Vec{e}(\Omega)$ and equal to $\frac{r}{T_c^2} \left| \Phi\left(
\frac{\Omega }{T_c}\right)\right|^2.$ Then, it is apparent that
the solution of the fixed point matrix equation
(\ref{fixed_point_intermediate_step}) is a matrix with the basis
of eigenvectors $\Mat{U}(\Omega),$  and
(\ref{fixed_point_intermediate_step}) reduces to the equation
corresponding to the $\left(\left \lfloor \frac{r-1}{2} \right
\rfloor + 1 \right)^{\mathrm{th}}$ element ${\upsilon}(\Omega)$ of
$\Mat{\Upsilon}(\Omega)=\Mat{U}^H(\Omega)\widetilde{\Mat{\Upsilon}}(\Omega)\Mat{U}(\Omega)$
\begin{equation}\label{fp_B}
\upsilon_{s}^{-1}(\Omega)=\sigma^2 + \beta \frac{r}{T_c^2} \left|
\Phi\left(\frac{\Omega}{T_c}  \right)\right|^2 \int \frac{\lambda
\mathrm{d}F_{|\Mat{A}|^2} (\lambda)}{1+ \frac{\lambda r}{2 \pi T_c^2}
\int_{-\pi}^{\pi} \left| \Phi\left(\frac{\Omega^{\prime}}{T_c}  \right)\right|^2 \upsilon(\Omega^{\prime}) \mathrm{d}\Omega^{\prime} }.
\end{equation}
The other components of the diagonal matrix $\Mat{\Upsilon}(\Omega)$ are simply given by $\upsilon_s^{-1}(\Omega)=\sigma^2, \, s=1, \ldots,r$ and $s \neq \left(\left
\lfloor \frac{r-1}{2} \right \rfloor + 1 \right).$  The identity $\Vec{e}^H(\Omega)
\widetilde{\Mat{\Upsilon}}(\Omega)  \Vec{e}(\Omega)=\upsilon(\Omega)$ yields
\begin{equation}\label{SINR_aux_final}
\lim_{K=\beta N \rightarrow \infty} \mathrm{SINR}_k=
\frac{|a_{k}|^2 r}{ 2 \pi T_c^2} \int_{-\pi}^{\pi}
\left| \Phi\left(\frac{\Omega}{T_c}  \right)\right|^2 \upsilon(\Omega)
\mathrm{d}\Omega  .
\end{equation}
The convergence (\ref{SINR_aux_final}) in probability or in the
first mean can be proven as in Theorem
\ref{theo:sinr_MMSE_chip_asynch}. By substituting (\ref{defMUE})
in (\ref{fp_B}) and (\ref{SINR_aux_final}), using definition (\ref{efficiency_spectrum}), and $\omega=\frac{\Omega}{T_c},$ we obtain the fixed point equation (\ref{fixed_point_eq_theo_small_bandwidth}) and the limit (\ref{SINR_theo_small_bandwidth}), respectively.

This concludes the proof of Theorem
\ref{theo:small_bandwidth_MMSE}.

\section{Proof of Proposition \ref{prop2}}\label{section:proof_proposition_2}
Proposition \ref{prop2} follows immediately from Corollary
\ref{cor_sinc_pulse}. In fact, from (\ref{fix_point_aux}) it is
apparent that the multiuser efficiency of a system with load
$\beta$ and sinc pulses having roll-off equal to $\alpha$ is equal
to the multiuser efficiency of a system with load
$\frac{\beta}{\alpha}$ and sinc pulses having zero roll-off.
Thanks to the fundamental relations between multiuser efficiency
and capacity \cite{guo:05} we obtain (\ref{eq:sinc_capacity}).
Since the spectral efficiency is obtained as the ratio $
\dfrac{\left.\mathcal{C}^{\text{(sinc)}} (\beta, \text{SNR},
\alpha)\right|_{\text{SNR}=N_0^{-1}}}{\alpha}, $ it is apparent
from (\ref{eq:sinc_capacity}) that it is constant as $\beta
\rightarrow \infty$ for any finite bandwidth $\alpha.$

\section{Proof of Corollary \ref{theor:constrained_capacity}}\label{section:proof_theor_constrained_capacity}

Let $\mathrm{MMSE}_{b_{k}[m]}(\rho)$ be the achievable MMSE by an
estimator of the symbol $b_{k}[m]$ transmitted by user $k$ in the
$m$-th symbol interval when the transmitted signal
$\vecb$ in (\ref{matrix_model}) is Gaussian and the signal to noise ratio is
$\rho=\sigma^{-2}.$ Furthermore, let
$\mathrm{SINR}_{b_{k}[m]}(\rho)$ be the SINR at the output of the
same MMSE estimator for the transmitted symbol $b_{k}[m].$ Then,
\begin{equation}\label{SNIR_MMSE_fundamental}
  \mathrm{MMSE}_{b_{k}[m]}(\rho)=
  \frac{1}{1+\mathrm{SINR}_{b_k[m]}(\rho)}.
\end{equation}
Additionally, let
$I(\vecb;\vecy, \rho)$ be
the mutual information in nats between the input
$\vecb$ and the output
$\vecy.$ From Theorem 2 in \cite{guo:05} the
following relation holds
\begin{equation}\label{MMSE_mutual_information}
  \frac{\mathrm{d}}{\mathrm{d}\rho}I(\vecb;\vecy,
  \rho)= \mathrm{E}\{\|\matH \vecb - \matH \widehat{\vecb}\|^2\}
\end{equation}
being $\widehat{\vecb}$ the conditional mean estimate. We
recall here that for Gaussian signals conditional mean estimate
and MMSE estimate coincide (see e.g., \cite{Kaybook}) and
\begin{align}
  \mathrm{E}\{\|\matH \vecb - \matH \widehat{\vecb}\|^2\} & = \mathrm{tr}\left( \sigma^2 \matH^H (\matH \matH^H+ \sigma^2 \matI)^{-1} \matH \right) \nonumber \\
  &= \sum_{m,k} \frac{\mathrm{SINR}_{b_k[m]}(\rho)}{\rho (1+ \mathrm{SINR}_{b_k[m]}(\rho))}.
\end{align}

For $K,N,m \rightarrow \infty$ with $\frac{K}{N} \rightarrow
\beta,$ $\mathrm{SINR}_{b_k[m]}(\rho)$ converges with probability
one to deterministic values. More specifically,
$\text{SINR}_{b_k[m]}(\rho)=\left. \frac{|a_{k}|^2 E_{\phi}}{N_0}
\eta_{ \gamma}\right|_{\gamma=\frac{r \rho }{T_c}},$ being $\eta_{\gamma}$ the multiuser efficiency corresponding to $\gamma=\frac{E_{\phi}}{N_0}$  as in Corollary
\ref{cor_MMSE_raised_cosine} or Theorem
\ref{theo:small_bandwidth_MMSE}.

Then,  the total capacity per chip constrained to a given chip
pulse waveform is given by

\begin{align}
\mathcal{C}^{(\text{asyn})}\left(\beta, \frac{E_{\phi}}{N_0}, \phi
\right) &= \frac{\beta}{\ln2}   \int_{0}^{\frac{E_{\phi} T_c}{r N_0}}  \int_0^{+\infty} \frac{\lambda \frac{r s}{T_c} \eta_{\frac{r s}{T_c}}\mathrm{d}F_{|\matA|^2 }(\lambda) \, \mathrm{d}s}{1+\lambda \frac{r s}{T_c} \eta_{\frac{r s}{T_c}}}\\
& = \frac{\beta }{\ln2}  \int_{0}^{\frac{E_{\phi}}{N_0}} \mathrm{d}t \int_0^{+\infty} \frac{\lambda \eta_t \mathrm{d}F_{|\matA|^2}(\lambda)}{1+ \lambda t \eta_t }
\end{align}

This concludes the proof of Corollary
\ref{theor:constrained_capacity}.

\section{Proof of Theorem \ref{theo:equivalence_widetildestackH}}\label{section:proof_theo_equivalence_widetildestackH}
In this section $\Mat{0}_{N}^{M}$ denotes an $N \times M $ matrix of zeros. Shortly, $\Mat{0}_{N}$ and $\Vec{e}_N$ denote $N$-dimensional vectors of zeros and ones, respectively. Additionally, we introduce the notation $\Mat{F}_{N}^{(r)}=\Mat{F}_{N} \otimes \Mat{I}_r,$ with $\Mat{F}_N$ already defined in (\ref{fourier_matrix}).

A basic property of the Fourier eigenbasis functions is stated in the following.
\begin{property}\label{prop_basic_fourier}
Let $\Mat{E}_L$ denote an $L \times L$ matrix of ones  and by $\boldsymbol{\mathcal{F}}_{L}^{(r)}(u),$ with $u \leq (L-1)Nr$, the $LNr \times LN $ matrix with structure
\begin{equation*}
 \boldsymbol{\mathcal{F}}_{L}^{(r)}(u) = \left( \begin{array}{ccc}
                                                  \Mat{0}_{ur}^{ur} & \Mat{0}_{ur}^{Nr} & \Mat{0}_{ur}^{(L-1)N-u} \\
                                                  \Mat{0}_{Nr}^{ur} & \Mat{F}_{N}^{(r)} & \Mat{0}_{Nr}^{(L-1)N-u} \\
                                                  \Mat{0}^{ur}_{(L-1)N-u} & \Mat{0}^{Nr}_{(L-1)N-u} & \Mat{0}_{(L-1)N-u}^{(L-1)N-u}
                                                \end{array}
  \right)
\end{equation*}
Then,
\begin{equation}\label{basic_prop_fourier}
    \frac{1}{\sqrt{L}} (\Mat{E}_L \otimes \Mat{F}_N^{(r)})\boldsymbol{\mathcal{F}}_{L}^{(r)}(u)=\boldsymbol{\mathcal{E}}(u,r,N)
\end{equation}
where $\boldsymbol{\mathcal{E}}(u,r,N)$ is a matrix with structure
\begin{equation*}
\boldsymbol{\mathcal{E}}(u,r,N)= \left( \begin{array}{c|c|c}
                                                   & \boldsymbol{\mathcal{R}}_{rN}^u & \\
                                                  \Mat{0}_{LNr}^{ur} & \vdots & \Mat{0}_{(L-1)Nr-ur}^{(L-1)Nr-ur} \\
                                                   & \boldsymbol{\mathcal{R}}_{rN}^u &
                                                \end{array}
  \right)
\end{equation*}
and $\boldsymbol{\mathcal{R}}_{rN}^u$ is an $rN \times N$ block diagonal matrix with $\ell$-th block $(\boldsymbol{\mathcal{R}}_{rN}^u )_{\ell \ell} = \mathrm{e}^{-j\frac{2 \pi}{N}(\ell -1) u} \Vec{e}_r. $
\end{property}

Let us consider the virtual spreading sequence of user $k$ for symbol $m$ in the time interval $[-MT_s, MT_s],$ with $M>m$ integer:
\begin{equation*}
    \pmb{\mathfrak{h}}_k^{(m)T}=a_k[\Mat{0}^{T}_{(M+m)Nr}, \widetilde{\boldsymbol{\Phi}}_k \vecs_k^{(m)}, \Mat{0}^{T}_{(M-m-1)Nr} ].
\end{equation*}
Property \ref{prop_basic_fourier} and decomposition (\ref{decomposition}) yield
\begin{align}\label{equivalence}
 \frac{1}{\sqrt{2M+1}} (\Mat{E}_{2M+1} \otimes \Mat{F}_N^{(r)})\boldsymbol{\mathfrak{h}}_k^{(m)}&=\Vec{e}_{2M+1} \otimes (\boldsymbol{\mathcal{R}}_{rN}^{\overline{\tau}_k} \boldsymbol{\Delta}_{\phi,r}(\widetilde{\tau}_k) \widetilde{\vecs}_k) \nonumber \\
 & =\frac{1}{\sqrt{2M+1}} (\Mat{E}_{2M+1} \otimes \Mat{F}_N^{(r)})\widehat{\boldsymbol{\mathfrak{h}}}_k^{(m)}
\end{align}
with $\widehat{\boldsymbol{\mathfrak{h}}}_k^{(m)}=a_k[\Mat{0}_{(M+m)Nr}^T, (\Mat{F}_N^{(r)}\boldsymbol{\mathcal{R}}_{rN}^{\overline{\tau}_k} \boldsymbol{\Delta}_{\phi,r}(\widetilde{\tau}_k) \widetilde{\vecs}_k)^T, \Mat{0}_{(M-m)Nr}^T ]^T.$
Let us observe that the position of the nonzero elements does not depends anymore on $\overline{\tau}.$
Since the random entries of the vector $\widetilde{\vecs}_k$ are rotationally invariant, the rotation matrix $\boldsymbol{\mathcal{R}}_{rN}^{\overline{\tau}_k} $ can be absorbed in the random vector $\widetilde{\vecs}_k$ without change of the statistics and the vector $\widehat{\boldsymbol{\mathfrak{h}}}_k$  can be rewritten as
\begin{equation*}
  \widehat{\boldsymbol{\mathfrak{h}}}_k= a_k [\Vec{0}_{(M+m)Nr}, \Vec{h}_k^{(m)}, \Vec{0}_{(M-m)Nr}]^T
\end{equation*}
where $\Vec{h}_k^{(m)}=\Mat{F}_N^{(r)} \boldsymbol{\Delta}_{\phi,r}(\widetilde{\tau}_k)$ is the virtual spreading for the $m$-th transmitted symbol of user $k$ delayed by $\widetilde{\tau} \in [0, T_c).$ From the previous considerations it follows that the random  matrix $\boldsymbol{\mathcal{T}}=\boldsymbol{\mathcal{HH}}^H$ is unitarily equivalent to an infinite block diagonal matrix with blocks ${\matT}^{(m)}= {\matH}^{(m)}{\matH}^{(m)H}$ of dimension $rN \times rN$ being $\matH^{(m)}$ the matrix with $k$-th column equal to $\vech_k^{(m)},$ i.e., the transfer matrix of a symbol synchronous but chip asynchronous system with time delay $\{\widetilde{\tau}_1, \ldots,\widetilde{\tau}_K \}.$ Then, asymptotically for $K, N \rightarrow + \infty$ the eigenvalue distribution of the matrix $\matT^{(m)}$ equals the eigenvalue distribution of the matrix ${\boldsymbol{\mathcal{T}}}.$ The equivalence of the systems in terms of SINR at the output of a linear MMSE detector and in terms of capacity follows also from the fact that the SINR is invariant to any unitary transform of the system transfer matrix as already observed in the proof of Theorem \ref{theo:sinr_MMSE_chip_asynch}.

This concludes the proof of Theorem
\ref{theo:equivalence_widetildestackH}.

\bibliographystyle{ieeetran}


\begin{biography}[]
{Laura Cottatellucci} is currently working as assistant professor at the department of Mobile Communications at Eurecom, France. She received the degree in Electrical Engineering  and the PhD from University "La Sapienza", Italy in 1995 and from Technical University of Vienna, Austria in 2006, respectively. She worked in Telecom Italia from 1995 until 2000. From April 2000 to September 2005 she was Senior Research at ftw., Vienna, Austria in the group of information processing for wireless communications. From October 2005 to December 2005 she was research fellow on  ad-hoc networks at INRIA, Sophia Antipolis, France and guest researcher at Eurecom,  Sophia Antipolis, France. From January 2006 to November 2006, Dr. Cottatellucci was appointed research fellow at the Institute for Telecommunications Research, University of South Australia, Adelaide, Australia working on information theory for networks with uncertain topology.  Her research interests lie in the field of network information theory, communication theory, and signal processing for wireless communications.
\end{biography}

\begin{biography}[]
{Ralf R. M\"uller} (S'96-M'03-SM'05) was born in Schwabach, Germany, 1970. He received the Dipl.- Ing. and Dr.-Ing. degree with distinction from University of Erlangen-Nuremberg in 1996 and 1999, respectively. From 2000 to 2004, he directed a research group at Vienna Telecommunications Research Center in Vienna, Austria and taught as an adjunct professor at Vienna University of Technology. Since 2005 he has been a full professor at the Department of Electronics and Telecommunications at the Norwegian University of Science and Technology (NTNU) in Trondheim, Norway. He held visiting appointments at Princeton University, US, Institute Eurecom, France, University of Melbourne, Australia, University of Oulu, Finland, National University of Singapore, Babes-Bolyai University, Cluj-Napoca, Romania, Kyoto University, Japan, and University of Erlangen-Nuremberg, Germany. Dr. M\"uller received the Leonard G. Abraham Prize (jointly with Sergio Verd\'u) for the paper "Design and analysis of low-complexity interference mitigation on vector channels" from the IEEE Communications Society. He was presented awards for his dissertation "Power and bandwidth efficiency of multiuser systems with random spreading" by the Vodafone Foundation for Mobile Communications and the German Information Technology Society (ITG). Moreover, he received the ITG award for the paper "A random matrix model for communication via antenna arrays," as well as the Philipp-Reis Award (jointly with Robert Fischer). Dr. M\"uller served as an associate editor for the IEEE TRANSACTIONS ON INFORMATION THEORY from 2003 to 2006.
\end{biography}

\begin{biography}[]{
M{\'e}rouane Debbah} was born in Madrid, Spain. He entered the Ecole Normale Suprieure de Cachan (France) in 1996 where he received his M.Sc and Ph.D. degrees respectively in 1999 and 2002. From 1999 to 2002, he worked for Motorola Labs on Wireless Local Area Networks and prospective fourth generation systems. From 2002 until 2003, he was appointed Senior Researcher at the Vienna Research Center for Telecommunications (FTW) (Vienna, Austria) working on MIMO wireless channel modeling issues. From 2003 until 2007, he joined the Mobile Communications department of Eurecom (Sophia Antipolis, France) as an Assistant Professor. He is presently a Professor at Supelec (Gif-sur-Yvette, France), holder of the Alcatel-Lucent Chair on Flexible Radio. His research interests are in information theory, signal processing and wireless communications. M{\'e}rouane Debbah is the recipient of the "Mario Boella" prize award in 2005, the 2007 General Symposium IEEE GLOBECOM best paper award, the Wi-Opt 2009 best paper award as well as the Valuetools 2007,Valuetools 2008 and CrownCom2009 best student paper awards. He is a WWRF fellow.
\end{biography}
\end{document}